\documentclass[12pt]{article}

\usepackage[cp1251]{inputenc}
\usepackage[T2A]{fontenc}
\usepackage[english]{babel}
\usepackage{ alphalph,etoolbox }
\usepackage{amsthm, amsmath, amssymb, amsbsy, amscd, amsfonts, latexsym, euscript,epsf}

\newtheorem{thm}{Theorem}[section]
\usepackage{epic,eepic}
\usepackage{texdraw}
\usepackage[dvips]{color}
\usepackage{color}
\usepackage{ dsfont }
\usepackage{graphicx}
\usepackage{exscale,relsize}
\usepackage{authblk}
\usepackage{url}
\usepackage{enumerate}

\allowdisplaybreaks
\newtheorem{proposition}[thm]{Proposition}
\newtheorem{corollary}[thm]{Corollary}
\newtheorem{lemma}[thm]{Lemma}

\theoremstyle{definition}

\newtheorem{example}[thm]{Example}
\newtheorem{remark}[thm]{Remark}

\newcommand{\md}{\mathcal{M}}
\newcommand{\cl}{\colon}
\newcommand{\dft}{\mathrm{d}}
\newcommand{\qqquad}{\qquad\quad}
\newcommand{\fmp}{\mathbf{F}}
\newcommand{\zsp}{\mathbb{Z}_{>0}}
\newcommand{\ts}{\mathrm{T}}

\DeclareMathOperator{\mat}{Mat}
\DeclareMathOperator{\GL}{GL}
\DeclareMathOperator{\End}{End}
\DeclareMathOperator{\id}{Id}

\newcommand{\mm}{\mathbf{M}}
\newcommand{\ybm}{\mathbf{Y}}
\newcommand{\pw}{p}
\newcommand{\ct}{c}
\newcommand{\ah}{\mathbf{a}}
\newcommand{\bh}{\mathbf{b}}
\newcommand{\ch}{\mathbf{c}}

\newcommand{\gmd}{\mathbf{G}}

\newcommand{\lb}{\label}
\newcommand{\er}{\eqref}

\title{Tetrahedron maps, Yang--Baxter maps, and partial linearisations}
\date{}

\author{S. Igonin\thanks{s-igonin@yandex.ru}}
\author{V. Kolesov\thanks{vadik.kolesov2015@yandex.ru}}
\author{S. Konstantinou-Rizos\thanks{skonstantin84@gmail.com}}
\author{M.M. Preobrazhenskaia\thanks{rita.preo@gmail.com}}
\affil{Centre of Integrable Systems, P.G. Demidov Yaroslavl State University, Yaroslavl, Russia}

\patchcmd{\subequations}{\alph{equation}}{\alphalph{\value{equation}}}{}{}

\begin{document}

\maketitle

\begin{abstract}

We study tetrahedron maps, which are set-theoretical solutions to the Zamolodchikov tetrahedron equation, 
and Yang--Baxter maps, which are set-theoretical solutions to the quantum Yang--Baxter equation.

In particular, we clarify the structure of the nonlinear algebraic relations which define 
linear (parametric) tetrahedron maps (with nonlinear dependence on parameters), and we present
several transformations which allow one to obtain new such maps from known ones.
Furthermore, we prove that the differential of a (nonlinear) tetrahedron map on a manifold is a tetrahedron map as well.
Similar results on the differentials of Yang--Baxter 
and entwining Yang--Baxter maps are also presented.

Using the obtained general results, we construct new examples of (parametric) 
Yang--Baxter and tetrahedron maps. 
The considered examples include maps associated with integrable systems and matrix groups.
In particular, we obtain a parametric family of new linear tetrahedron maps, 
which are linear approximations for the nonlinear tetrahedron map
constructed by Dimakis and M\"uller-Hoissen~\cite{Dimakis}
in a study of soliton solutions of vector Kadomtsev--Petviashvili (KP) equations.
Also, we present invariants for this nonlinear tetrahedron map.

\end{abstract}

\bigskip

\hspace{.2cm} \textbf{PACS numbers:} 02.30.Ik, 02.90.+p, 03.65.Fd.

\hspace{.2cm} \textbf{Mathematics Subject Classification 2020:} 16T25, 81R12.


\hspace{.2cm} \textbf{Keywords:} 
Zamolodchikov tetrahedron equation, quantum Yang--Baxter equation, 

\hspace{2.5cm} parametric tetrahedron maps, parametric Yang--Baxter maps, entwining Yang--Baxter

\hspace{2.5cm} maps, linearisations, differentials of maps.

\section{Introduction}
\lb{sintr}

The Zamolodchikov tetrahedron equation~\cite{Zamolodchikov,Zamolodchikov-2}
is a higher-dimensional analogue of the well-celebrated quantum Yang--Baxter equation.
They belong to the most fundamental equations in mathematical physics 
and have applications in many diverse branches of physics and mathematics, 
including statistical mechanics, quantum field theories, algebraic topology, 
and the theory of integrable systems. 
Some applications of the tetrahedron equation can be found 
in~\cite{Bazhanov-Sergeev,Bazhanov-Sergeev-2,DMH15,Doliwa-Kashaev,Gorbounov-Talalaev,Kapranov,Kashaev-Sergeev,Kassotakis-Tetrahedron,Sharygin-Talalaev,Nijhoff-2,Nijhoff,Yoneyama21} and references therein.

This paper is devoted to tetrahedron maps and Yang--Baxter maps, 
which are set-theoretical solutions to the tetrahedron equation 
and the Yang--Baxter equation, respectively.
Set-theoretical solutions of the Yang--Baxter equation have been intensively studied
by many authors after the work of Drinfeld~\cite{Drin92}.
Even before that, examples of such solutions were constructed by Sklyanin~\cite{skl88}.
A quite general construction for tetrahedron maps first appeared 
in the work of Korepanov~\cite{Korepanov-DS} in connection with 
integrable dynamical systems in discrete time.
Presently, the relations of tetrahedron maps and Yang--Baxter maps
with integrable systems (including PDEs and lattice equations) is a very active area of research 
(see, e.g.,~\cite{abs2003,atk14,Vincent,Dimakis,Doliwa-Kashaev,fordy17,Kashaev-Sergeev,Pavlos-Maciej-2,Kassotakis-Tetrahedron,KR,Sokor-Sasha,Korepanov-DS,Kouloukas2,pap-Tongas,Yoneyama21} and references therein).

This paper is organised as follows.

Section~\ref{preliminaries} contains the definitions  
of (parametric) tetrahedron maps, (parametric) Yang--Baxter maps
and recalls some basic properties of them.

Sections~\ref{sltm} and~\ref{slptm} are devoted to linear tetrahedron maps 
and to linear parametric tetrahedron maps with nonlinear dependence on parameters.
In particular, we clarify the structure 
of the nonlinear algebraic relations that define such maps, 
and we present several transformations which allow one to obtain new such maps from known ones.

The results of Sections~\ref{sltm},~\ref{slptm} on linear (parametric) tetrahedron maps
generalise some results of~\cite{BIKRP} on linear (parametric) Yang--Baxter maps.
Examples of linear parametric Yang--Baxter maps related to integrable PDEs of 
vector Kadomtsev--Petviashvili (KP) and (deformed) nonlinear Schr\"odinger (NLS) types
were discussed also in~\cite{Dimakis,Sokor-Sasha}.

Remark~\ref{rlappr}, Corollary~\ref{ctaaa}, and Examples~\ref{eeltr},~\ref{edmh} show how 
linear tetrahedron maps appear as linear approximations of nonlinear ones.

Hietarinta~\cite{Hietarinta97} considered some special linear tetrahedron maps.
A relation of our results with those of~\cite{Hietarinta97} is discussed in Remark~\ref{rhiet}.

We study also nonlinear Yang--Baxter and tetrahedron maps on manifolds 
and their differentials defined on the corresponding tangent bundles.
For a manifold~$\md$, its tangent bundle is denoted by~$T\md$.
When we consider maps of manifolds, 
we assume that they are either smooth, or complex-analytic, or rational,
so that the differential is defined for such a map.
Section~\ref{sdybt} contains the following results: 
\begin{itemize}
	\item For any Yang--Baxter map $Y\cl\md\times \md\to \md\times\md$,
the differential $\dft Y\cl T\md\times T\md\to T\md\times T\md$
is a Yang--Baxter map of the manifold $T\md\times T\md$.
A similar result is valid also for entwining Yang--Baxter maps.
\item For any tetrahedron map 
$\mathbf{T}\cl\md\times\md\times\md\to\md\times\md\times\md$, 
the differential 
$$
\dft\mathbf{T}\cl T\md\times T\md\times T\md\to T\md\times T\md\times T\md
$$ 
is a tetrahedron map of the manifold $T\md\times T\md\times T\md$.
\end{itemize}
The above result on the differential of a Yang--Baxter map was used (without proof) in~\cite{BIKRP}.

Examples of the differentials for tetrahedron maps are presented in Section~\ref{sdybt}.
The computed differentials are tetrahedron maps~\er{difelt},~\er{difkm}.
An example of a computation of the differentials for a family of Yang--Baxter maps 
is given in Section~\ref{sYBmvb}.

In Example~\ref{edmh} we consider the nonlinear birational tetrahedron map~\er{dmt}
which was constructed by Dimakis and M\"uller-Hoissen~\cite{Dimakis}
in a study of soliton solutions of vector KP equations.
We present invariants for this map and find for it a linear approximation, 
which is a family of new linear tetrahedron maps~\er{pmat2} 
depending on the parameter~$\ct\in\mathbb{C}$.

Generalising some constructions from~\cite{BIKRP},
in Section~\ref{sYBmvb} we present new examples of linear 
parametric Yang--Baxter maps (with nonlinear dependence on parameters)
associated with some matrix groups.
Let $\mathbb{K}$ be either $\mathbb{C}$ or $\mathbb{R}$.
For $n\in\zsp$, consider the matrix group~$\GL_n(\mathbb{K})$ 
and an abelian subgroup $\Omega\subset\GL_n(\mathbb{K})$.
The construction in Section~\ref{sYBmvb} involves the computation 
of the differential of a nonlinear Yang--Baxter map 
on~$\GL_n(\mathbb{K})\times\GL_n(\mathbb{K})$ 
and restricting the computed differential to the submanifold
$$
\big(\Omega\times\mat_n(\mathbb{K})\big)\times\big(\Omega\times\mat_n(\mathbb{K})\big)
$$
of the tangent bundle 
$$
T\big(\GL_n(\mathbb{K})\times\GL_n(\mathbb{K})\big)\cong
T\GL_n(\mathbb{K})\times T\GL_n(\mathbb{K})\cong 
\big(\GL_n(\mathbb{K})\times\mat_n(\mathbb{K})\big)
\times\big(\GL_n(\mathbb{K})\times\mat_n(\mathbb{K})\big)
$$ 
of the manifold $\GL_n(\mathbb{K})\times\GL_n(\mathbb{K})$.
Furthermore, we extend the obtained Yang--Baxter map to 
a family of Yang--Baxter maps depending on a nonzero constant $l\in\mathbb{K}$.

As a result, for any nonzero $l\in\mathbb{K}$, $\,n,\pw\in\zsp$, 
and any abelian subgroup $\Omega\subset\GL_n(\mathbb{K})$,
we obtain the parametric Yang--Baxter map~\er{lpyblg} with parameters $a,b\in\Omega$.
For $\pw\ge 2$ the map~\er{lpyblg} is new, to our knowledge.
For $\pw=1$ it was presented in~\cite{BIKRP}.

In the construction of~\eqref{lpyblg} we assume
that $\mathbb{K}$ is either $\mathbb{R}$ or $\mathbb{C}$, 
in order to use tangent spaces and differentials.
Furthermore, one can verify that~\eqref{lpyblg} is a parametric Yang--Baxter map 
for any field~$\mathbb{K}$.

Section~\ref{sconc} concludes the paper 
with comments on how the results of this paper can be extended.

\begin{remark}
According to Remarks~\ref{rnonlpar},~\ref{rdl}
many constructions of this paper involve Yang--Baxter and tetrahedron maps which 
are ``partly linear'' in the sense that the maps are linear with respect to some of the variables 
and nonlinear with respect to the other variables.
Thus, informally speaking, one can say that we deal with ``partial linearisations'' of 
Yang--Baxter and tetrahedron maps.
\end{remark}

\section{Preliminaries}\label{preliminaries}
\subsection{Tetrahedron maps}

For any set $S$ and $n\in\zsp$, we use the notation
$S^n=\underbrace{S\times S\times\dots\times S}_{n}$.

Let $W$ be a set. A \emph{tetrahedron map} is a map 
\begin{equation}
\notag
T\cl W^3\rightarrow W^3,\qquad
T(x,y,z)=(u(x,y,z),v(x,y,z),w(x,y,z)),\qquad x,y,z\in W,
\end{equation}
satisfying the (Zamolodchikov) \emph{tetrahedron equation}
\begin{equation}\label{Tetrahedron-eq}
    T^{123}\circ T^{145} \circ T^{246}\circ T^{356}=T^{356}\circ T^{246}\circ T^{145}\circ T^{123}.
\end{equation}
Here $T^{ijk}\cl W^6\rightarrow W^6$ for $i,j,k=1,\ldots,6$, $i<j<k$, is the map 
acting as $T$ on the $i$th, $j$th, $k$th factors 
of the Cartesian product $W^6$ and acting as identity on the remaining factors.
For instance,
$$
T^{246}(x,y,z,r,s,t)=(x,u(y,r,t),z,v(y,r,t),s,w(y,r,t)),\qqquad x,y,z,r,s,t\in W.
$$
The schematic interpretation of the tetrahedron equation is given in Figure \ref{fig-tetra}. Every line with a number $i=1,\ldots,6$, corresponds to one of six copies of the set~$W$, 
and every intersection point of lines $i$, $j$, $k$ corresponds to the map~$T^{ijk}$.

\begin{figure}[ht]
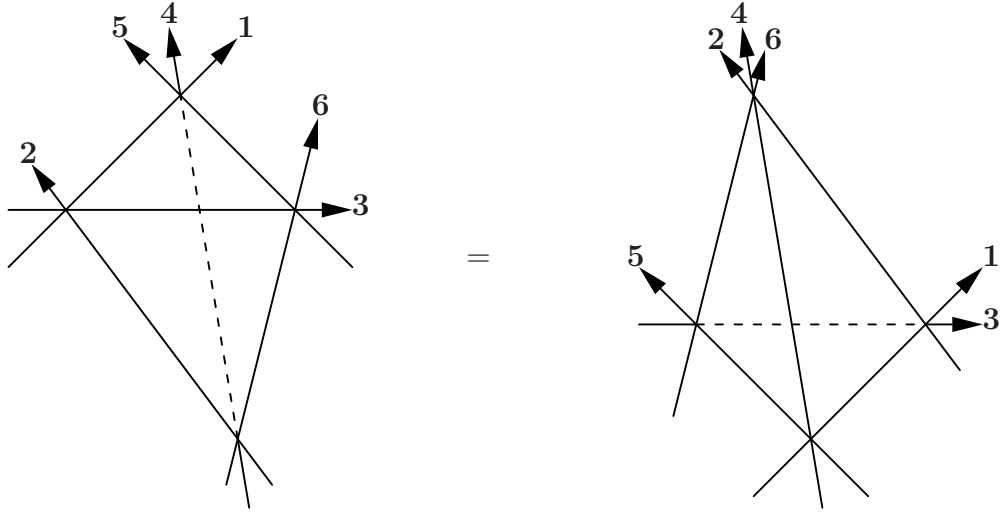

\centering
\centertexdraw{ 
\setunitscale 0.6
\move(-1.5 -0.5) \arrowheadtype t:F \avec(0.5 1.5)
\move(-1.5 0) \arrowheadtype t:F \avec(1.5 0)
\move(0.6 -2.6) \lvec(0.5 -2)
\move(0 1) \arrowheadtype t:F \avec(-0.1 1.6)
\lpatt(0.067 0.1) \move(0.5 -2) \lvec(0 1) 
\lpatt()
\move(1.5 -0.5) \arrowheadtype t:F \avec(-0.5 1.5)
\move(0.8 -2.4) \arrowheadtype t:F \avec(-1.3 0.4)
\move(0.4 -2.4) \arrowheadtype t:F \avec(1.2 0.8)

\htext (2.5 -.5) {$=$}
\htext (0.5 1.5) {\textbf{1}}
\htext (1.5 -.05) {\textbf{3}}
\htext (-0.17 1.62) {\textbf{4}}
\htext (-0.6 1.5) {\textbf{5}}
\htext (-1.4 0.4) {\textbf{2}}
\htext (1.15 0.82) {\textbf{6}}

\move(5 -2.5) \arrowheadtype t:F \avec(7 -0.5)
\move(4 -1) \lvec(4.5 -1) 
\lpatt(0.067 0.1)
\move(4.5 -1) \lvec(6.5 -1)
\lpatt()
\move(6.5 -1) \arrowheadtype t:F \avec(7 -1)
\move(5.6 -2.6) \arrowheadtype t:F \avec(4.9 1.6)
\move(6 -2.5) \arrowheadtype t:F \avec(4 -.5)
\move(6.8 -1.4) \arrowheadtype t:F \avec(4.7 1.4)
\move(4.3 -1.8) \arrowheadtype t:F \avec(5.1 1.4)

\htext (7 -0.5) {\textbf{1}}
\htext (7 -1.05) {\textbf{3}}
\htext (4.8 1.62) {\textbf{4}}
\htext (3.9 -.5) {\textbf{5}}
\htext (4.6 1.4) {\textbf{2}}
\htext (5.1 1.4) {\textbf{6}}
}

\caption{Schematic interpretation of the tetrahedron 
equation~\cite{Zamolodchikov-2, Kapranov, Kassotakis-Tetrahedron}.}\label{fig-tetra}
\end{figure}
\begin{proposition}[\cite{Kassotakis-Tetrahedron}]
\lb{pp13}
Consider the permutation map
\begin{gather*}
P^{13}\cl W^3\to W^3,\qqquad
P^{13}(a_1,a_2,a_3)=(a_3,a_2,a_1),\qquad a_i\in W.
\end{gather*}
If a map $T\cl W^3\to W^3$ satisfies the tetrahedron equation~\er{Tetrahedron-eq}
then $\tilde{T}=P^{13}\circ T\circ P^{13}$ obeys this equation as well.
\end{proposition}

\begin{proposition}[\cite{Kassotakis-Tetrahedron}]
\lb{pist}
Let $T\cl W^3\to W^3$ be a tetrahedron map.
Suppose that a map $\sigma\cl W\to W$ satisfies 
\begin{gather*}
(\sigma\times\sigma\times\sigma)\circ T\circ(\sigma\times\sigma\times\sigma)=T,
\qqquad \sigma\circ\sigma=\id.
\end{gather*}
Then
\begin{gather*}
\tilde{T}=(\sigma\times\id\times\sigma)\circ T\circ(\id\times\sigma\times\id),\qqquad
\hat{T}=(\id\times\sigma\times\id)\circ T\circ(\sigma\times\id\times\sigma)
\end{gather*}
are tetrahedron maps.
\end{proposition}

\subsection{Parametric tetrahedron maps}

Let $\Omega$ and $V$ be sets.
Here $\Omega$ is regarded as a set of parameters.
Consider a map of the form 
\begin{gather}
\lb{tomv}
T\colon (\Omega\times V)\times (\Omega\times V)\times (\Omega\times V)
\to
(\Omega\times V)\times (\Omega\times V)\times (\Omega\times V),\\
\notag
((\alpha,x),(\beta,y),(\gamma,z))\mapsto 
\left((\alpha,u((\alpha,x),(\beta,y),(\gamma,z))),(\beta,v((\alpha,x),(\beta,y),(\gamma,z))),
(\gamma,w((\alpha,x),(\beta,y),(\gamma,z)))\right),\\
\notag
\alpha,\beta,\gamma\in\Omega,\qquad x,y,z\in V,\qquad
u((\alpha,x),(\beta,y),(\gamma,z)),v((\alpha,x),(\beta,y),(\gamma,z)),w((\alpha,x),(\beta,y),(\gamma,z))\in V.
\end{gather}
Note that the map~\er{tomv} satisfies $\pi\circ T=\pi$ for the projection 
$\pi\colon (\Omega\times V)\times(\Omega\times V)\times(\Omega\times V)
\to\Omega\times\Omega\times\Omega$.

We set 
\begin{gather}
\lb{uvabg}
u_{\alpha,\beta,\gamma}(x,y,z)=u((\alpha,x),(\beta,y),(\gamma,z)),\qqquad 
v_{\alpha,\beta,\gamma}(x,y,z)=v((\alpha,x),(\beta,y),(\gamma,z))),\\
\lb{wabg}
w_{\alpha,\beta,\gamma}(x,y,z))=w((\alpha,x),(\beta,y),(\gamma,z)),\qqquad
\alpha,\beta,\gamma\in\Omega,\qquad x,y,z\in V.
\end{gather}

For the map~\er{tomv}, equation~\eqref{Tetrahedron-eq} with $W=\Omega\times V$ can be written as
\begin{equation}\label{Par-Tetrahedron-eq}
    T^{123}_{\alpha,\beta,\gamma}\circ T^{145}_{\alpha,\delta,\epsilon} 
		\circ T^{246}_{\beta,\delta,\zeta}\circ T^{356}_{\gamma,\epsilon,\zeta}=
		T^{356}_{\gamma,\epsilon,\zeta}\circ T^{246}_{\beta,\delta,\zeta}\circ 
		T^{145}_{\alpha,\delta,\epsilon}\circ T^{123}_{\alpha,\beta,\gamma}
		\qquad\text{for all }\,\alpha,\beta,\gamma,\delta,\epsilon,\zeta\in\Omega.
\end{equation}
The terms $T^{123}_{\alpha,\beta,\gamma}$, $T^{145}_{\alpha,\delta,\epsilon}$, 
$T^{246}_{\beta,\delta,\zeta}$, $T^{356}_{\gamma,\epsilon,\zeta}$ in~\er{Par-Tetrahedron-eq}
are maps $V^6\to V^6$ defined similarly to the terms in equation~\er{Tetrahedron-eq}, 
adding the parameters $\alpha$, $\beta$, $\gamma$, $\delta$, $\epsilon$, $\zeta$.
For instance,
$$
T^{246}_{\beta,\delta,\zeta}(x,y,z,r,s,t)=
(x,u_{\beta,\delta,\zeta}(y,r,t),z,v_{\beta,\delta,\zeta}(y,r,t),s,w_{\beta,\delta,\zeta}(y,r,t)),
\qqquad x,y,z,r,s,t\in V.
$$

We use the notation
\begin{gather}
\lb{gtabg}
T_{\alpha,\beta,\gamma}\cl V^3\to V^3,\qquad
T_{\alpha,\beta,\gamma}(x,y,z)=(u_{\alpha,\beta,\gamma}(x,y,z),v_{\alpha,\beta,\gamma}(x,y,z),w_{\alpha,\beta,\gamma}(x,y,z)),\\
\notag
\alpha,\beta,\gamma\in\Omega,\qquad x,y,z\in V.
\end{gather}
Thus, $T_{\alpha,\beta,\gamma}$ defined by~\er{gtabg} is a map $V^3\to V^3$ 
depending on parameters $\alpha,\beta,\gamma\in\Omega$.

Equation~\eqref{Par-Tetrahedron-eq} is called the 
\emph{parametric tetrahedron equation}.
The family of maps~\er{gtabg} 
is called a \emph{parametric tetrahedron map} if it satisfies 
equation~\eqref{Par-Tetrahedron-eq}.
Then we can say more briefly that $T_{\alpha,\beta,\gamma}$ is a parametric tetrahedron map.

\begin{remark}
\label{rnonlpar}
Thus, the parametric map $T_{\alpha,\beta,\gamma}$ 
defined by~\er{gtabg},~\er{uvabg},~\er{wabg} obeys 
the parametric tetrahedron equation~\er{Par-Tetrahedron-eq} 
if and only if the (nonparametric) map~\er{tomv} obeys
the tetrahedron equation~\er{Tetrahedron-eq}.

In Section~\ref{slptm} we consider the case when $V$ is a vector space
and for any values of $\alpha,\beta,\gamma$
the map $T_{\alpha,\beta,\gamma}\cl V^3\to V^3$ is linear.
Note that usually $\Omega$ is a subset of another vector space, 
and the dependence of~$T_{\alpha \beta \gamma}$ on the 
parameters $\alpha,\beta,\gamma\in\Omega$ is nonlinear.

Then one can say that in Section~\ref{slptm} we study 
tetrahedron maps of the form~\er{tomv} which are linear with respect to~$V$ 
and may be nonlinear with respect to~$\Omega$.
However, it is useful to keep $\alpha,\beta,\gamma$ 
as parameters and to work with~$T_{\alpha\beta\gamma}$ 
instead of~$T$ from~\er{tomv}.
\end{remark}

\subsection{Tetrahedron maps vs Yang--Baxter maps}

In this subsection, 
we recall some simple relations between (parametric) tetrahedron maps 
and (parametric) Yang--Baxter maps, which are defined below.
The results of Proposition~\ref{prop_YB_tetr} and Corollary~\ref{cptyb} 
are known and are very simple, but for completeness we present proofs for them. 
Proposition~\ref{prop_YB_tetr} and Corollary~\ref{cptyb} show 
that (parametric) Yang--Baxter maps can be regarded as a particular class 
of (parametric) tetrahedron maps.

Let $W$ be a set. A \emph{Yang--Baxter map} is a map 
$$
Y\colon W\times W\to W\times W,\qquad Y(x,y)=(u(x,y),v(x,y)),\qquad x,y\in W,
$$
satisfying the Yang--Baxter equation
\begin{equation}\label{eq_YB}
Y^{12}\circ Y^{13}\circ Y^{23}=Y^{23}\circ Y^{13}\circ Y^{12}.
\end{equation}
The terms $Y^{12}$, $Y^{13}$, $Y^{23}$ in~\eqref{eq_YB} are maps $W^3\to W^3$ defined as follows 
\begin{gather*}
Y^{12}(x,y,z)=\big(u(x,y),v(x,y),z\big),\qquad\quad
Y^{23}(x,y,z)=\big(x,u(y,z),v(y,z)\big),\\
Y^{13}(x,y,z)=\big(u(x,z),y,v(x,z)\big),\qquad\quad x,y,z\in W.
\end{gather*}

\begin{proposition}\label{prop_YB_tetr}
Let $Y\colon W^2\rightarrow W^2$ be a Yang--Baxter map. 
Then the maps
$$
Y^{23}\colon W^3\rightarrow W^3,\qqquad
Y^{12}\colon W^3\rightarrow W^3
$$ 
are tetrahedron maps.
\end{proposition}
\begin{proof}
Let $T=Y^{23}$. We need to prove~\eqref{Tetrahedron-eq}.
Using the identity map $\id_{W^3}\cl W^3\to W^3$, one obtains
\begin{gather}
\notag
T^{123}=Y^{23}\times\id_{W^3},\qquad
T^{145}=\id_{W^3}\times Y^{12},\qquad
T^{246}=\id_{W^3}\times Y^{13},\qquad
T^{356}=\id_{W^3}\times Y^{23},\\
\lb{tt16}
T^{123}\circ T^{145} \circ T^{246}\circ T^{356}=
(Y^{23}\times\id_{W^3})\circ(\id_{W^3}\times(Y^{12}\circ Y^{13}\circ Y^{23})),\\
\lb{tt33}
\begin{gathered}
T^{356}\circ T^{246}\circ T^{145}\circ T^{123}=
(\id_{W^3}\times(Y^{23}\circ Y^{13}\circ Y^{12}))\circ(Y^{23}\times\id_{W^3})=\\
=(Y^{23}\times\id_{W^3})\circ(\id_{W^3}\times(Y^{23}\circ Y^{13}\circ Y^{12})).
\end{gathered}
\end{gather}
Since $Y$ satisfies \eqref{eq_YB}, from~\er{tt16},~\er{tt33} we derive~\eqref{Tetrahedron-eq} 
for $T=Y^{23}$. Similarly, one can prove~\eqref{Tetrahedron-eq} for $T=Y^{12}$.
\end{proof}

Let $\Omega$ and $V$ be sets. 
A \emph{parametric Yang--Baxter map} $Y_{\alpha,\beta}$ is a family of maps
\begin{gather}
\label{ParamYB}
Y_{\alpha,\beta}\colon V\times V\to V\times V,\qquad
Y_{\alpha,\beta}(x,y)=\big(u_{\alpha,\beta}(x,y),\,v_{\alpha,\beta}(x,y)\big),\qquad x,y\in V,
\qquad \alpha,\beta\in\Omega,
\end{gather}
depending on parameters $\alpha,\beta\in\Omega$ 
and satisfying the parametric Yang--Baxter equation
\begin{gather}
\label{pybeq}
Y^{12}_{\alpha,\beta}\circ Y^{13}_{\alpha,\gamma} \circ Y^{23}_{\beta,\gamma}=
Y^{23}_{\beta,\gamma}\circ Y^{13}_{\alpha,\gamma} \circ Y^{12}_{\alpha,\beta}\qquad\text{for all }\,\alpha,\beta,\gamma\in\Omega.
\end{gather}
The terms $Y^{12}_{\alpha,\beta}$, $Y^{13}_{\alpha,\gamma}$, $Y^{23}_{\beta,\gamma}$
in~\eqref{pybeq} are maps $V^3\to V^3$ given by
\begin{gather*}
Y^{12}_{\alpha,\beta}(x,y,z)=\big(u_{\alpha,\beta}(x,y),v_{\alpha,\beta}(x,y),z\big),\qquad\quad
Y^{23}_{\beta,\gamma}(x,y,z)=\big(x,u_{\beta,\gamma}(y,z),v_{\beta,\gamma}(y,z)\big),\\
Y^{13}_{\alpha,\gamma}(x,y,z)=\big(u_{\alpha,\gamma}(x,z),y,v_{\alpha,\gamma}(x,z)\big),\qqquad
x,y,z\in V.
\end{gather*}
A parametric Yang--Baxter map~\eqref{ParamYB} with parameters $\alpha,\beta$
can be interpreted as the following Yang--Baxter map~$\ybm$ without parameters
\begin{gather}
\label{ypsv}
\ybm\colon (\Omega\times V)\times (\Omega\times V)\to(\Omega\times V)\times (\Omega\times V),\quad
\ybm\big((\alpha,x),(\beta,y)\big)=\big((\alpha,u_{\alpha,\beta}(x,y)),
(\beta,v_{\alpha,\beta}(x,y))\big).
\end{gather}

\begin{corollary}\label{cptyb}
For any parametric Yang--Baxter map~\eqref{ParamYB}, the maps 
\begin{gather}
\lb{ttyy}
T_{\alpha,\beta,\gamma}=Y^{23}_{\beta,\gamma}\cl V^3\to V^3,\qqquad
\tilde{T}_{\alpha,\beta,\gamma}=Y^{12}_{\alpha,\beta}\cl V^3\to V^3
\end{gather}
are parametric tetrahedron maps.
\end{corollary}
\begin{proof}
Applying Proposition~\ref{prop_YB_tetr} to the (nonparametric) 
Yang--Baxter map $\ybm$ given by~\er{ypsv}, we see that $\ybm^{23}$, $\ybm^{12}$ 
are (nonparametric) tetrahedron maps. 
This is equivalent to the fact that \er{ttyy} are parametric tetrahedron maps.
\end{proof}

\section{Linear tetrahedron maps}
\label{sltm}

For any vector space $W$ we denote by $\End(W)$ 
the set of linear maps $W\to W$.

Let $V$ be a vector space over a field $\mathbb{K}$. 
Usually $\mathbb{K}$ is either $\mathbb{C}$ or $\mathbb{R}$. 
In this section we consider linear maps 
$T\colon V^3 \to V^3$ given by
\begin{equation}\label{tetrahedron-linear}
T\cl \begin{pmatrix}
x\\
y\\
z
\end{pmatrix} \mapsto \begin{pmatrix}
u\\
v\\
w
\end{pmatrix} = \begin{pmatrix}
     A&B&C  \\
     D&E&F  \\
     K&L&M
\end{pmatrix}
\begin{pmatrix}
x\\
y\\
z
\end{pmatrix},  \qquad x,y,z,u,v,w \in V,
\end{equation}
where $A,B,C,D,E,F,K,L,M \in \End(V)$.

\begin{remark}
If $V=\mathbb{K}^n$ for some $n\in\zsp$ then 
$A,B,C,D,E,F,K,L,M$ are $n\times n$ matrices.
\end{remark}

\begin{proposition}\label{Lin-Tetrahedron-Eqs-proposition}
A map $T \in \End(V^3)$ given by \eqref{tetrahedron-linear} satisfies 
the tetrahedron equation~\eqref{Tetrahedron-eq} if and only 
if the maps $A,B,C,D,E,F,K,L,M\in \End(V)$ in~\eqref{tetrahedron-linear} obey the following equations
\begin{subequations}\label{Linear-Tetrahedron-eqs}
    \begin{eqnarray}
    & DA=AD+BDA, \quad AB=BA+ABD, \quad  ED=DE+EDB, \quad  BE=EB+DBE,  &\label{LTE1} \\
    & LE=EL+FLE, \quad EF=FE+EFL, \quad ML=MLF+LM , \quad FM=MF+LFM, & \\
    & D=DD+EDA, \quad B=BB+ABE, \quad L=LL+MLE, \quad F=FF+EFM, &\label{LTE3} \\ 
    & K=DK+EKA+FKD+FLDA, \quad C=CB+ACE+BCL+ABFL, & \\
    & KA=AK+BKA+CKD+CLDA, \quad AC=CA+ACD+BCK+ABFK, & \\
    & KK+LKA+MKD+MLDA=0, \quad CC+ACF+BCM+ABFM=0, & \\
    & EA+DBD=AE+BDB, \quad ME+LFL=EM+FLF, & \\
    & K=KL+LKB+MKE+MLDB, \quad C=FC+DCF+ECM+DBFM, & \\
    & MK=KM+LKC+MKF+MLDC, \quad CM=MC+KCF+LCM+KBFM, & \\
    & FD+EFK=DF+EDC, \quad BL+CLE=LB+KBE, & \\
    & LD=DL+EKB+FKE+FLDB, \quad BF=FB+DCE+ECL+DBFL, & \\
    & LA+KBD=AL+BKB+CKE+CLDB, & \\
    & AF+BDC=FA+DCD+ECK+DBFK, & \\
    & MD+LFK=DM+EKC+FKF+FLDC, & \\
    & BM+CLF=MB+KCE+LCL+KBFL, & \\
    & MA-AM+KCD+LCK+KBFK=BKC+CKF+CLDC. &
    \end{eqnarray}
\end{subequations}
\end{proposition}
\begin{proof}
This can be proved by substitution of $T$ in \eqref{tetrahedron-linear} 
to the tetrahedron equation~\eqref{Tetrahedron-eq}.
For any $x,y,z,r,s,t\in V$, from the left-hand side of~\eqref{Tetrahedron-eq} we obtain
\begin{eqnarray*}
    &(T^{123}\circ T^{145} \circ T^{246}\circ T^{356})(x,y,z,r,s,t)=(AAx+ABDy+ABEr+ABFKz+ABFLs+\\
    &+ABFMt+ACEs+ACFt+BAy+BBr+BCKz+BCLs+CAz+CBs+CCt,\\
    &DAx+BDBy+DBEr+DBFKz+DBFLs+DBFMt+DCDz+DCEs+DCFt+\\
    &+EAy+EBr+ECKz+ECLs+ECMt+FAz+FBs+FCt,\\
    &KAx+KBDy+KBEr+KBFKz+KBFLs+KBFMt+KCDz+KCEs+\\
    &+KCFt+LAy+LBr+LCKz+LCLs+LCMt+MAz+MBs+MCt,\\
    &Dx+EDy+EEr+EFKz+EFLs+EFMt+FDz+FEs+\\
    &+FFt,Kx+LDy+LEr+LFKz+LFLs+LFMt+MDz+MEs+MFt,\\
    &Ky+Lr+MKz+MLs+MMt),
\end{eqnarray*}
while the right-hand side of \eqref{Tetrahedron-eq} implies
\begin{eqnarray*}
     &(T^{356}\circ T^{246}\circ T^{145}\circ T^{123})(x,y,z,r,s,t)=(AAx+ABy+ACz+Br+Cs,\\
     &ADx+AEy+AFz+BDAx+BDBy+BDCz+BEr+BFs+Ct,\\
     &AKx+ALy+AMz+BKAx+BKBy+BKCz+BLr+BMs+CKDx+\\
     &+CKEy+CKFz+CLDAx+CLDBy+CLDCz+CLEr+CLFs+CMt,\\
     &DDx+DEy+DFz+EDAx+EDBy+EDCz+EEr+EFs+Ft,\\
     &DKx+DLy+DMz+EKAx+EKBy+EKCz+ELr+EMs+FKDx+FKEy+\\
     &+FKFz+FLDAx+FLDBy+FLDCz+FLEr,FLFs+FMt,\\
     &KKx+KLy+KMz+LKAx+LKBy+LKCz+LLr+LMs+MKDx+MKEy+\\
     &+MKFz+MLDAx+MLDBy+MLDCz+MLEr+MLFs+MMt).
\end{eqnarray*}
By equating the coefficients of $x,y,z,r,s,t$ for each component of these vectors, 
we derive a system of relations equivalent to \eqref{Tetrahedron-eq}. 
For instance, consider the coefficients of~$y$ in the first components:
\[ABD + BA = AB.\]
This is the second equation from~\eqref{LTE1}.
Performing the same actions with all variables and components of the obtained vectors, 
we get all of relations~\eqref{Linear-Tetrahedron-eqs}.
\end{proof}

\begin{corollary}
System \eqref{Linear-Tetrahedron-eqs} implies the matrix equations
\eqref{lin-tet-matrix-eq1}--\eqref{lin-tet-matrix-eq4}
\begin{subequations}
\lb{mteq}
\begin{align}
    &\begin{pmatrix}
        D & E  \\
        A & B
    \end{pmatrix}
    \begin{pmatrix}
        D & BE  \\
        DA & B
    \end{pmatrix}=
    \begin{pmatrix}
        D & BE  \\
        DA & B
    \end{pmatrix},\label{lin-tet-matrix-eq1}\\
    &\begin{pmatrix}
        L & M  \\
        E & F
    \end{pmatrix}
    \begin{pmatrix}
        L & FM  \\
        LE & F
    \end{pmatrix}=
   \begin{pmatrix}
        L & FM  \\
        LE & F
    \end{pmatrix},\\
    &\begin{pmatrix}
        AB & B  \\
        D & ED
    \end{pmatrix}
    \begin{pmatrix}
        D & E  \\
        A & B
    \end{pmatrix}=
    \begin{pmatrix}
        AB & B  \\
        D & ED
    \end{pmatrix},\\
    &\begin{pmatrix}
        EF & F  \\
        L & ML
    \end{pmatrix}
    \begin{pmatrix}
        L & M  \\
        E & F
    \end{pmatrix}=
    \begin{pmatrix}
         EF & F  \\
         L & ML
    \end{pmatrix},\label{lin-tet-matrix-eq4}
\end{align}
\end{subequations}
as well as the following
\begin{subequations} 
\notag
\begin{align}
    [E-BED,A]=[BD,D+B],& \qquad [A-DAB,D]=[DB,D+B], \\
    [M-FML,E]=[FM,M+F],& \qquad  [E-LEF,M]=[MF,M+F], \\
    [B+D-DB,E]=0,& \qquad [B+D-BD,A]=0, \\
    [L+F-LF,M]=0,& \qquad  [L+F-FL,E]=0, \\
    [E,FK+CL-KB-DC] &+ [F+L,D+B]+[DB,FL]=0,\label{lin-tet-commutator-eq5}
\end{align}
\end{subequations}
where by $\left[\cdot,\cdot\right]$ we denote the commutator $\left[A,B\right]=AB-BA$.
\end{corollary}
\begin{remark}
\label{requiv}
Equations~\eqref{mteq} are equivalent to~\eqref{LTE1}--\eqref{LTE3}.
Thus, equations~\eqref{LTE1}--\eqref{LTE3}
can be replaced by equations~\eqref{mteq}, which have more clear structure.
\end{remark}

\begin{proposition}
\lb{pinvt}
For any vector space $V$,
the set of linear tetrahedron maps \eqref{tetrahedron-linear} is invariant 
under the following transformations
\begin{gather}
\lb{tcsym}
\begin{pmatrix}
         A&B&C  \\
         D&E&F  \\
         K&L&M
\end{pmatrix}\mapsto
\begin{pmatrix}
         M&L&K  \\
         F&E&D  \\
         C&B&A
\end{pmatrix},\\
\lb{tchsn}
\begin{pmatrix}
         A&B&C  \\
         D&E&F  \\
         K&L&M
\end{pmatrix}\mapsto
\begin{pmatrix}
         -A&B&-C  \\
         D&-E&F  \\
         -K&L&-M
\end{pmatrix}.
\end{gather}

Let $V=\mathbb{K}^n$ for some $n\in\zsp$.
Then $A,B,C,D,E,F,K,L,M$ in~\eqref{tetrahedron-linear} are $n\times n$ matrices.
In this case, the set of linear tetrahedron maps~\eqref{tetrahedron-linear} 
is invariant also under the transformation
\begin{gather}
\lb{trtr}
\begin{pmatrix}
         A&B&C  \\
         D&E&F  \\
         K&L&M
\end{pmatrix}\mapsto
\begin{pmatrix}
         A&B&C  \\
         D&E&F  \\
         K&L&M
\end{pmatrix}^\ts=
\begin{pmatrix}
         A^\ts&D^\ts&K^\ts  \\
         B^\ts&E^\ts&L^\ts  \\
         C^\ts&F^\ts&M^\ts
 \end{pmatrix},
\end{gather}
where $\ts$ denotes the transpose operation for matrices.
\end{proposition}
\begin{proof}
The statement about the transformation~\er{tcsym} 
follows from Proposition~\ref{pp13} with $W=V$.

The case of the transformation~\er{tchsn} 
follows from Proposition~\ref{pist},
if we take $W=V$ and consider the map $\sigma\cl V\to V$, $\sigma(v)=-v$.

To prove the statement about the transformation~\er{trtr},
one can apply the transpose operation 
to both sides of the tetrahedron equation~\er{Tetrahedron-eq}
for $T$ given by~\er{tetrahedron-linear}.
\end{proof}

\begin{example}
\lb{eltmp}
Let $V=\mathbb{K}^2$. Then $V^3=\mathbb{K}^6$.
Let $\ct\in\mathbb{K}$, $\ct\neq 0$.
Consider the linear map $T\in\End(V^3)$ given by~\er{tetrahedron-linear}
with the matrix
\begin{gather}
\lb{abcp}
\begin{pmatrix}
     A&B&C  \\
     D&E&F  \\
     K&L&M
\end{pmatrix}=
\begin{pmatrix} 
1 & 0 & \frac{\ct-1}{\ct} & \frac{(\ct-1)^2 (\ct+1)}{\ct} & -\frac{\ct-1}{\ct} & 1-\ct \\
 0 & 1 & 0 & 1-\ct & 0 & 0 \\
 0 & 0 & \frac{1}{\ct} & -\frac{(\ct-1)^2 (\ct+1)}{\ct} & \frac{\ct-1}{\ct} & \ct-1 \\
 0 & 0 & 0 & \ct & 0 & 0 \\
 0 & 0 & 0 & 1-\ct & 1 & 0 \\
 0 & 0 & 0 & 1-\ct & 0 & 1
\end{pmatrix}.
\end{gather}
Thus, for~\er{abcp} we have
\begin{gather}
\lb{pabc}
\begin{gathered}
A=\begin{pmatrix} 1 & 0 \\ 0 & 1\end{pmatrix},\quad
B=\begin{pmatrix} \frac{\ct-1}{\ct} & \frac{(\ct-1)^2 (\ct+1)}{\ct} \\ 0 & 1-\ct\end{pmatrix},\quad
C=\begin{pmatrix} -\frac{\ct-1}{\ct} & 1-\ct \\ 0 & 0\end{pmatrix},\quad
D=\begin{pmatrix} 0 & 0 \\ 0 & 0\end{pmatrix},\\
E=\begin{pmatrix} \frac{1}{\ct} & -\frac{(\ct-1)^2 (\ct+1)}{\ct} \\ 0 & \ct\end{pmatrix},\,\,\,\ 
F=\begin{pmatrix} \frac{\ct-1}{\ct} & \ct-1 \\ 0 & 0\end{pmatrix},\,\,\,\ 
K=\begin{pmatrix} 0 & 0 \\ 0 & 0\end{pmatrix},\,\,\,\ 
L=\begin{pmatrix} 0 & 1-\ct \\ 0 & 1-\ct\end{pmatrix},\,\,\,\ 
M=\begin{pmatrix} 1 & 0 \\ 0 & 1 \end{pmatrix}.
\end{gathered}
\end{gather}
Using Proposition~\ref{Lin-Tetrahedron-Eqs-proposition},
one can verify that \er{abcp} is a tetrahedron map. 

As explained in Example~\ref{edmh}, we have derived this linear tetrahedron map, 
using the differential of a nonlinear tetrahedron map from~\cite{Dimakis}.
Applying the transformations~\er{tcsym}, \er{tchsn}, \er{trtr} 
and their compositions to~\er{abcp}, we obtain several more linear tetrahedron maps.
\end{example}

\begin{proposition}\label{par-family}
Let $T_1$, $T_2$ be linear tetrahedron maps of the form
\begin{equation*}
    T_1=\begin{pmatrix}
     A&B&0  \\
     D&E&0  \\
     0&0&M
\end{pmatrix},\qqquad
    T_2=\begin{pmatrix}
     \tilde{A}&0&0  \\
     0&\tilde{E}&\tilde{F}\\
     0&\tilde{L}&\tilde{M}
\end{pmatrix}.
\end{equation*}
Let $l,m\in\mathbb{K}$, $l\neq 0$. Then 
\begin{equation*}
    T_1^{l,m}=\begin{pmatrix}
     lA&B&0  \\
     D&l^{-1}E&0  \\
     0&0&mM
\end{pmatrix},\qqquad
    T_2^{l,m}=\begin{pmatrix}
     m\tilde{A}&0&0  \\
     0&l\tilde{E}&\tilde{F}\\
     0&\tilde{L}&l^{-1}\tilde{M}
\end{pmatrix}
\end{equation*}
are linear tetrahedron maps as well.
\end{proposition}
\begin{proof}
For each $i=1,2$, the fact that $T_i$ obeys equations~\er{Linear-Tetrahedron-eqs} 
implies that $T_i^{l,m}$ obeys these equations as well.
\end{proof}

\begin{remark}
\lb{rhiet}
Hietarinta~\cite{Hietarinta97} studied some special linear tetrahedron maps.
In our notation, Hietarinta~\cite{Hietarinta97} assumes that 
$A,B,C,D,E,F,K,L,M$ in~\er{tetrahedron-linear} belong to a commutative ring.
The assumption that $A,B,C,D,E,F,K,L,M$ commute simplifies 
equations~\er{Linear-Tetrahedron-eqs} very considerably, 
and this simplified version of~\er{Linear-Tetrahedron-eqs} appears in~\cite{Hietarinta97}.

Results of~\cite{Hietarinta97} imply that for any $\ah,\bh,\ch\in\mathbb{K}$
the following matrix determines a linear tetrahedron map 
$\mathbb{K}^3\to\mathbb{K}^3$
\begin{gather}
\lb{hietm}
\begin{pmatrix}
     A&B&C  \\
     D&E&F  \\
     K&L&M
\end{pmatrix}=
\begin{pmatrix}
     \ah&1-\ah\bh&0  \\
     0&\bh&0  \\
     0&1-\bh\ch&\ch
\end{pmatrix}
\end{gather}
Hietarinta~\cite{Hietarinta97} studied the case when 
$\ah$, $\bh$, $\ch$ are elements of a commutative ring.
\end{remark}

\section{Linear parametric tetrahedron maps}
\label{slptm}

Let $V$ be a vector space over a field $\mathbb{K}$. Let $\Omega$ be a set. 
In this section we study linear maps $T_{\alpha\beta\gamma} \in \End(V^3)$ 
depending on parameters $\alpha,\beta,\gamma \in \Omega$. 
We consider a linear map
\begin{equation}
\lb{tabg}
T_{\alpha \beta \gamma}\cl \begin{pmatrix}
    x\\
    y\\
    z
\end{pmatrix} \mapsto \begin{pmatrix}
    u\\
    v\\
    w
\end{pmatrix} = \begin{pmatrix}
     A_{\alpha \beta \gamma}&B_{\alpha \beta \gamma}&C_{\alpha \beta \gamma}  \\
     D_{\alpha \beta \gamma}&E_{\alpha \beta \gamma}&F_{\alpha \beta \gamma}  \\
     K_{\alpha \beta \gamma}&L_{\alpha \beta \gamma}&M_{\alpha \beta \gamma}
\end{pmatrix}
\begin{pmatrix}
    x\\
    y\\
    z
\end{pmatrix}, \qquad x,y,z,u,v,w \in V,
\end{equation}
where $A_{\alpha \beta \gamma},B_{\alpha \beta \gamma},C_{\alpha \beta \gamma},D_{\alpha \beta \gamma},E_{\alpha \beta \gamma},F_{\alpha \beta \gamma},K_{\alpha \beta \gamma},L_{\alpha \beta \gamma},M_{\alpha \beta \gamma} \in \End(V)$ for all $\alpha,\beta,\gamma \in \Omega$.
Then $T_{\alpha \beta \gamma}$ is called a \emph{linear parametric tetrahedron map} 
if it satisfies the parametric tetrahedron equation~\eqref{Par-Tetrahedron-eq}.
\begin{remark}
\lb{dpnl}
Note that usually $\Omega$ is a subset of another vector space, 
and the dependence of~$T_{\alpha \beta \gamma}$ on the 
parameters $\alpha$, $\beta$, $\gamma$ is nonlinear.
Examples of such maps are presented in Section~\ref{sYBmvb}.
\end{remark}
\begin{proposition}\label{Lin-Par-Tetrahedron-Eqs-proposition}
A parametric map $T_{\alpha \beta \gamma}$ given by~\er{tabg}
satisfies the parametric tetrahedron equation~\eqref{Par-Tetrahedron-eq} 
if and only if it obeys the following list of equations 
for all values of the parameters $\alpha,\beta,\gamma,\delta,\epsilon,\zeta\in\Omega$
\begin{subequations}\label{parametric-linear-tetrahedron-eqs}
\begin{align}
 & A_{\alpha \beta \gamma}A_{\alpha \delta \epsilon}= 
 A_{\alpha \delta \epsilon}A_{\alpha \beta \gamma}, \quad
E_{\alpha \delta \epsilon}E_{\beta \delta \zeta}= 
 E_{\beta \delta \zeta}E_{\alpha \delta \epsilon}, \quad
M_{\beta \delta \zeta}M_{\gamma \epsilon \zeta}= 
 M_{\gamma \epsilon \zeta}M_{\beta \delta \zeta},
 \\ &
D_{\alpha \beta \gamma}A_{\alpha \delta \epsilon}= 
 A_{\beta \delta \zeta}D_{\alpha \beta \gamma}+ 
  B_{\beta \delta \zeta}D_{\alpha \delta \epsilon}A_{\alpha \beta \gamma},\quad 
 A_{\alpha \delta \epsilon}B_{\alpha \beta \gamma}= B_{\alpha \beta \gamma}A_{\beta \delta \zeta}+ 
  A_{\alpha \beta \gamma}B_{\alpha \delta \epsilon}D_{\beta \delta \zeta},\label{Lin-Par-eq-da-ad}
  \\ &
E_{\alpha \delta \epsilon}D_{\beta \delta \zeta}= D_{\beta \delta \zeta}E_{\alpha \beta \gamma}+ 
  E_{\beta \delta \zeta}D_{\alpha \delta \epsilon}B_{\alpha \beta \gamma},\quad
B_{\beta \delta \zeta}E_{\alpha \delta \epsilon}= 
 E_{\alpha \beta \gamma}B_{\beta \delta \zeta}+ 
  D_{\alpha \beta \gamma}B_{\alpha \delta \epsilon}E_{\beta \delta \zeta}, \label{Lin-Par-eq-be-eb}
   \\ &
 L_{\alpha \delta \epsilon}E_{\beta \delta \zeta}= E_{\gamma \epsilon \zeta}L_{\alpha \delta \epsilon}+ 
  F_{\gamma \epsilon \zeta}L_{\beta \delta \zeta}E_{\alpha \delta \epsilon},\quad
 E_{\beta \delta \zeta}F_{\alpha \delta \epsilon}=F_{\alpha \delta \epsilon}E_{\gamma \epsilon \zeta}+ 
  E_{\alpha \delta \epsilon}F_{\beta \delta \zeta}L_{\gamma \epsilon \zeta},
 \\ &
M_{\beta \delta \zeta}L_{\gamma \epsilon \zeta}= 
 M_{\gamma \epsilon \zeta}L_{\beta \delta \zeta}F_{\alpha \delta \epsilon}+ 
  L_{\gamma \epsilon \zeta}M_{\alpha \delta \epsilon}, \quad
F_{\gamma \epsilon \zeta}M_{\beta \delta \zeta}= 
 M_{\alpha \delta \epsilon}F_{\gamma \epsilon \zeta}+ 
  L_{\alpha \delta \epsilon}F_{\beta \delta \zeta}M_{\gamma \epsilon \zeta},
 \\&
D_{\alpha \delta \epsilon}= D_{\beta \delta \zeta}D_{\alpha \beta \gamma}+ 
  E_{\beta \delta \zeta}D_{\alpha \delta \epsilon}A_{\alpha \beta \gamma}, \quad B_{\alpha \delta \epsilon} =
  B_{\alpha \beta \gamma}B_{\beta \delta \zeta}+
 A_{\alpha \beta \gamma}B_{\alpha \delta \epsilon}E_{\beta \delta \zeta},\label{Lin-Par-eq-d-dd}
\\ &
L_{\beta \delta \zeta}= 
 L_{\gamma \epsilon \zeta}L_{\alpha \delta \epsilon}+ 
  M_{\gamma \epsilon \zeta}L_{\beta \delta \zeta}E_{\alpha \delta \epsilon}, \quad
F_{\beta \delta \zeta}= 
 F_{\alpha \delta \epsilon}F_{\gamma \epsilon \zeta}+ 
  E_{\alpha \delta \epsilon}F_{\beta \delta \zeta}M_{\gamma \epsilon \zeta},\label{Par-Lin-eq-l-ll}
 \\ &
K_{\alpha \delta \epsilon}= 
 D_{\gamma \epsilon \zeta}K_{\alpha \beta \gamma}+ 
  E_{\gamma \epsilon \zeta}K_{\alpha \delta \epsilon}A_{\alpha \beta \gamma}
  +F_{\gamma \epsilon \zeta}K_{\beta \delta \zeta}D_{\alpha \beta \gamma}+ 
  F_{\gamma \epsilon \zeta}L_{\beta \delta \zeta}D_{\alpha \delta \epsilon}A_{\alpha \beta \gamma},
   \\ &
 C_{\alpha \delta \epsilon}=C_{\alpha \beta \gamma}B_{\gamma \epsilon \zeta}+ 
  A_{\alpha \beta \gamma}C_{\alpha \delta \epsilon}E_{\gamma \epsilon \zeta}+ 
  B_{\alpha \beta \gamma}C_{\beta \delta \zeta}L_{\gamma \epsilon \zeta}+
  A_{\alpha \beta \gamma}B_{\alpha \delta \epsilon}F_{\beta \delta \zeta}L_{\gamma \epsilon \zeta},
\\ &
K_{\alpha \beta \gamma}A_{\alpha \delta \epsilon}= 
 A_{\gamma \epsilon \zeta}K_{\alpha \beta \gamma}+ 
  B_{\gamma \epsilon \zeta}K_{\alpha \delta \epsilon}A_{\alpha \beta \gamma}
  +C_{\gamma \epsilon \zeta}K_{\beta \delta \zeta}D_{\alpha \beta \gamma}+ 
  C_{\gamma \epsilon \zeta}L_{\beta \delta \zeta}D_{\alpha \delta \epsilon}A_{\alpha \beta \gamma},
  \\ &
A_{\alpha \delta \epsilon}C_{\alpha \beta \gamma}= 
  C_{\alpha \beta \gamma}A_{\gamma \epsilon \zeta}+ 
  A_{\alpha \beta \gamma}C_{\alpha \delta \epsilon}D_{\gamma \epsilon \zeta}+ 
  B_{\alpha \beta \gamma}C_{\beta \delta \zeta}K_{\gamma \epsilon \zeta}+ 
  A_{\alpha \beta \gamma}B_{\alpha \delta \epsilon}F_{\beta \delta \zeta}K_{\gamma \epsilon \zeta},
  \\ &
K_{\gamma \epsilon \zeta}K_{\alpha \beta \gamma}+ 
  L_{\gamma \epsilon \zeta}K_{\alpha \delta \epsilon}A_{\alpha \beta \gamma}+ 
  M_{\gamma \epsilon \zeta}K_{\beta \delta \zeta}D_{\alpha \beta \gamma}
  +M_{\gamma \epsilon \zeta}L_{\beta \delta \zeta}D_{\alpha \delta \epsilon}A_{\alpha \beta \gamma}= 0,
  \\ &
C_{\alpha \beta \gamma}C_{\gamma \epsilon \zeta}+ 
  A_{\alpha \beta \gamma}C_{\alpha \delta \epsilon}F_{\gamma \epsilon \zeta}+ 
  B_{\alpha \beta \gamma}C_{\beta \delta \zeta}M_{\gamma \epsilon \zeta}+
  A_{\alpha \beta \gamma}B_{\alpha \delta \epsilon}F_{\beta \delta \zeta}M_{\gamma \epsilon \zeta}= 0,
\\ &
E_{\alpha \beta \gamma}A_{\beta \delta \zeta}+ 
  D_{\alpha \beta \gamma}B_{\alpha \delta \epsilon}D_{\beta \delta \zeta}= 
 A_{\beta \delta \zeta}E_{\alpha \beta \gamma}+ 
  B_{\beta \delta \zeta}D_{\alpha \delta \epsilon}B_{\alpha \beta \gamma},
  \\ &
M_{\alpha \delta \epsilon}E_{\gamma \epsilon \zeta}+ 
  L_{\alpha \delta \epsilon}F_{\beta \delta \zeta}L_{\gamma \epsilon \zeta}= 
 E_{\gamma \epsilon \zeta}M_{\alpha \delta \epsilon}+ 
  F_{\gamma \epsilon \zeta}L_{\beta \delta \zeta}F_{\alpha \delta \epsilon},
 \\ &
K_{\beta \delta \zeta}= 
 K_{\gamma \epsilon \zeta}L_{\alpha \beta \gamma}+ 
  L_{\gamma \epsilon \zeta}K_{\alpha \delta \epsilon}B_{\alpha \beta \gamma}
 + M_{\gamma \epsilon \zeta}K_{\beta \delta \zeta}E_{\alpha \beta \gamma}+ 
  M_{\gamma \epsilon \zeta}L_{\beta \delta \zeta}D_{\alpha \delta \epsilon}B_{\alpha \beta \gamma},
  \\ &
C_{\beta \delta \zeta}= 
 F_{\alpha \beta \gamma}C_{\gamma \epsilon \zeta}+ 
  D_{\alpha \beta \gamma}C_{\alpha \delta \epsilon}F_{\gamma \epsilon \zeta}+ 
  E_{\alpha \beta \gamma}C_{\beta \delta \zeta}M_{\gamma \epsilon \zeta}+
  D_{\alpha \beta \gamma}B_{\alpha \delta \epsilon}F_{\beta \delta \zeta}M_{\gamma \epsilon \zeta},
 \\ &
M_{\beta \delta \zeta}K_{\gamma \epsilon \zeta}=K_{\gamma \epsilon \zeta}M_{\alpha \beta \gamma}+ 
  L_{\gamma \epsilon \zeta}K_{\alpha \delta \epsilon}C_{\alpha \beta \gamma}+
  M_{\gamma \epsilon \zeta}K_{\beta \delta \zeta}F_{\alpha \beta \gamma}+ 
  M_{\gamma \epsilon \zeta}L_{\beta \delta \zeta}D_{\alpha \delta \epsilon}C_{\alpha \beta \gamma},
 \\ &
C_{\gamma \epsilon \zeta}M_{\beta \delta \zeta}= 
 M_{\alpha \beta \gamma}C_{\gamma \epsilon \zeta}+ 
  K_{\alpha \beta \gamma}C_{\alpha \delta \epsilon}F_{\gamma \epsilon \zeta}+ 
  L_{\alpha \beta \gamma}C_{\beta \delta \zeta}M_{\gamma \epsilon \zeta}+
  K_{\alpha \beta \gamma}B_{\alpha \delta \epsilon}F_{\beta \delta \zeta}M_{\gamma \epsilon \zeta},
 \\ &
F_{\alpha \delta \epsilon}D_{\gamma \epsilon \zeta}+ 
  E_{\alpha \delta \epsilon}F_{\beta \delta \zeta}K_{\gamma \epsilon \zeta}= 
 D_{\beta \delta \zeta}F_{\alpha \beta \gamma}+ 
  E_{\beta \delta \zeta}D_{\alpha \delta \epsilon}C_{\alpha \beta \gamma},
 \\ &
B_{\gamma \epsilon \zeta}L_{\alpha \delta \epsilon}+ 
  C_{\gamma \epsilon \zeta}L_{\beta \delta \zeta}E_{\alpha \delta \epsilon}= 
 L_{\alpha \beta \gamma}B_{\beta \delta \zeta}+ 
  K_{\alpha \beta \gamma}B_{\alpha \delta \epsilon}E_{\beta \delta \zeta},
  \\ &
L_{\alpha \delta \epsilon}D_{\beta \delta \zeta}=
 D_{\gamma \epsilon \zeta}L_{\alpha \beta \gamma}+ 
  E_{\gamma \epsilon \zeta}K_{\alpha \delta \epsilon}B_{\alpha \beta \gamma}
  +F_{\gamma \epsilon \zeta}K_{\beta \delta \zeta}E_{\alpha \beta \gamma}+ 
  F_{\gamma \epsilon \zeta}L_{\beta \delta \zeta}D_{\alpha \delta \epsilon}B_{\alpha \beta \gamma},
  \\ &
 B_{\beta \delta \zeta}F_{\alpha \delta \epsilon}=F_{\alpha \beta \gamma}B_{\gamma \epsilon \zeta}+ 
  D_{\alpha \beta \gamma}C_{\alpha \delta \epsilon}E_{\gamma \epsilon \zeta}+ 
  E_{\alpha \beta \gamma}C_{\beta \delta \zeta}L_{\gamma \epsilon \zeta}+ 
  D_{\alpha \beta \gamma}B_{\alpha \delta \epsilon}F_{\beta \delta \zeta}L_{\gamma \epsilon \zeta},
  \\ & 
 L_{\alpha  \beta  \gamma}A_{\beta  \delta  \zeta}+ 
  K_{\alpha  \beta  \gamma}B_{\alpha  \delta  \epsilon}D_{\beta  \delta  \zeta}= 
 A_{\gamma  \epsilon  \zeta}L_{\alpha  \beta  \gamma}+ 
  B_{\gamma  \epsilon  \zeta}K_{\alpha  \delta  \epsilon}B_{\alpha  \beta  \gamma}
  +C_{\gamma  \epsilon  \zeta}K_{\beta  \delta  \zeta}E_{\alpha  \beta  \gamma}+ 
  C_{\gamma  \epsilon  \zeta}L_{\beta  \delta  \zeta}D_{\alpha  \delta  \epsilon}B_{\alpha  \beta  \gamma},
 \\ &
 A_{\beta \delta \zeta}F_{\alpha \beta \gamma}+ 
  B_{\beta \delta \zeta}D_{\alpha \delta \epsilon}C_{\alpha \beta \gamma}=F_{\alpha \beta \gamma}A_{\gamma \epsilon \zeta}+ 
  D_{\alpha \beta \gamma}C_{\alpha \delta \epsilon}D_{\gamma \epsilon \zeta}+ 
  E_{\alpha \beta \gamma}C_{\beta \delta \zeta}K_{\gamma \epsilon \zeta}+
  D_{\alpha \beta \gamma}B_{\alpha \delta \epsilon}F_{\beta \delta \zeta}K_{\gamma \epsilon \zeta}
  \\ &
M_{\alpha \delta \epsilon}D_{\gamma \epsilon \zeta} +
   L_{\alpha \delta \epsilon}F_{\beta \delta \zeta}K_{\gamma \epsilon \zeta} = 
 D_{\gamma \epsilon \zeta}M_{\alpha \beta \gamma}+ 
  E_{\gamma \epsilon \zeta}K_{\alpha \delta \epsilon}C_{\alpha \beta \gamma}+
  F_{\gamma \epsilon \zeta}K_{\beta \delta \zeta}F_{\alpha \beta \gamma}+ 
  F_{\gamma \epsilon \zeta}L_{\beta \delta \zeta}D_{\alpha \delta \epsilon}C_{\alpha \beta \gamma},
 \\ &
 B_{\gamma \epsilon \zeta}M_{\alpha \delta \epsilon}+ 
  C_{\gamma \epsilon \zeta}L_{\beta \delta \zeta}F_{\alpha \delta \epsilon}=M_{\alpha \beta \gamma}B_{\gamma \epsilon \zeta}+ 
  K_{\alpha \beta \gamma}C_{\alpha \delta \epsilon}E_{\gamma \epsilon \zeta}+ 
  L_{\alpha \beta \gamma}C_{\beta \delta \zeta}L_{\gamma \epsilon \zeta}+ 
  K_{\alpha \beta \gamma}B_{\alpha \delta \epsilon}F_{\beta \delta \zeta}L_{\gamma \epsilon \zeta},
 \\ &
  M_{\alpha \beta \gamma}A_{\gamma \epsilon \zeta}-
  A_{\gamma \epsilon \zeta}M_{\alpha \beta \gamma}+ 
  K_{\alpha \beta \gamma}C_{\alpha \delta \epsilon}D_{\gamma \epsilon \zeta}+ 
  L_{\alpha \beta \gamma}C_{\beta \delta \zeta}K_{\gamma \epsilon \zeta}+
  K_{\alpha \beta \gamma}B_{\alpha \delta \epsilon}F_{\beta \delta \zeta}K_{\gamma \epsilon \zeta}= \nonumber \\
 &=B_{\gamma \epsilon \zeta}K_{\alpha \delta \epsilon}C_{\alpha \beta \gamma}+ 
  C_{\gamma \epsilon \zeta}K_{\beta \delta \zeta}F_{\alpha \beta \gamma}+
  C_{\gamma \epsilon \zeta}L_{\beta \delta \zeta}D_{\alpha \delta \epsilon}C_{\alpha \beta \gamma}.
\end{align} 
\end{subequations}
\end{proposition}
\begin{proof}
The proof is similar to the proof of Proposition \ref{Lin-Tetrahedron-Eqs-proposition}.
\end{proof}

\begin{remark}
\lb{rperm}
In what follows we deduce some consequences from~\er{parametric-linear-tetrahedron-eqs}.
Note that, since equations~\er{parametric-linear-tetrahedron-eqs} must hold 
for all values of the parameters $\alpha,\beta,\gamma,\delta,\epsilon,\zeta\in\Omega$, 
we are allowed to make any permutation of the parameters in these equations.
\end{remark}

\begin{proposition} 
System  \eqref{parametric-linear-tetrahedron-eqs} implies the following matrix equations
\begin{subequations}
\lb{pmteq}
\begin{eqnarray}
    &\begin{pmatrix}
        D_{\alpha \beta \gamma} & E_{\alpha \beta \gamma}  \\
        A_{\alpha \beta \gamma} & B_{\alpha \beta \gamma}
    \end{pmatrix}
    \begin{pmatrix}
        D_{\zeta \alpha \epsilon} & B_{\alpha \delta \epsilon}E_{\beta \delta \zeta}  \\
        D_{\zeta \beta \delta}A_{\zeta \alpha \epsilon} & B_{\beta \delta \zeta}
    \end{pmatrix}=
    \begin{pmatrix}
        D_{\zeta \beta \delta} & B_{\beta \delta \zeta}E_{\alpha \delta \epsilon}  \\
        D_{\zeta \alpha \epsilon}A_{\zeta \beta \delta} & B_{\alpha \delta \epsilon}
    \end{pmatrix}, \label{Lin-Par-eq-matrix-1}
\\
    &\begin{pmatrix}
        A_{\zeta \alpha \epsilon}B_{\zeta \beta \delta} & B_{\zeta \alpha \epsilon}  \\
        D_{\beta \delta \zeta} & E_{\beta \delta \zeta}D_{\alpha \delta \epsilon}
    \end{pmatrix}
    \begin{pmatrix}
        D_{\alpha \beta \gamma} & E_{\alpha \beta \gamma}  \\
        A_{\alpha \beta \gamma} & B_{\alpha \beta \gamma}
    \end{pmatrix}=
    \begin{pmatrix}
        A_{\zeta \beta \delta}B_{\zeta \alpha \epsilon} & B_{\zeta \beta \delta}  \\
        D_{\alpha \delta \epsilon} & E_{\alpha \delta \epsilon}D_{\beta \delta \zeta}
    \end{pmatrix},  \label{Lin-Par-eq-matrix-2}
\\
    &\begin{pmatrix}
        L_{\alpha \delta \epsilon} & M_{\alpha \delta \epsilon}  \\
        E_{\alpha \delta \epsilon} & F_{\alpha \delta \epsilon}
    \end{pmatrix}
    \begin{pmatrix}
        L_{\zeta \beta \delta} & F_{\beta \delta \zeta}M_{\gamma \epsilon \zeta}  \\
        L_{\gamma \beta \epsilon}E_{\zeta \beta \delta} & F_{\gamma \epsilon \zeta}
    \end{pmatrix}=
    \begin{pmatrix}
        L_{\gamma \beta \epsilon} & F_{\gamma \epsilon \zeta}M_{\beta \delta \zeta}  \\
        L_{\zeta \beta \delta}E_{\gamma \beta \epsilon} & F_{\beta \delta \zeta}
    \end{pmatrix},  \label{Lin-Par-eq-matrix-3}
\\
    &\begin{pmatrix}
        E_{\zeta \beta \delta}F_{\gamma \beta \epsilon} & F_{\zeta \beta \delta}  \\
        L_{\gamma \epsilon \zeta} & M_{\gamma \epsilon \zeta}L_{\beta \delta \zeta}
    \end{pmatrix}
    \begin{pmatrix}
        L_{\alpha \delta \epsilon} & M_{\alpha \delta \epsilon}  \\
        E_{\alpha \delta \epsilon} & F_{\alpha \delta \epsilon}
    \end{pmatrix}=
    \begin{pmatrix}
        E_{\gamma \beta \epsilon}F_{\zeta \beta \delta} & F_{\gamma \beta \epsilon} \\
        L_{\beta \delta \zeta} & M_{\beta \delta \zeta}L_{\gamma \epsilon \zeta}
    \end{pmatrix}.  \label{Lin-Par-eq-matrix-4}
\end{eqnarray}
\end{subequations}
\end{proposition}
\begin{proof}
As explained in Remark~\ref{rperm}, 
we are allowed to make any permutation of the parameters.
Making the permutation 
$(\alpha,\beta,\gamma ,\delta, \epsilon, \zeta) \rightarrow (\zeta,\alpha,\epsilon, \beta, \delta, \gamma)$
in the first of~\eqref{Lin-Par-eq-da-ad} and the first of~\eqref{Lin-Par-eq-d-dd} and 
taking the second from~\eqref{Lin-Par-eq-be-eb} and the second from~\eqref{Lin-Par-eq-d-dd},
we obtain~\eqref{Lin-Par-eq-matrix-1}. 
Equations~\er{Lin-Par-eq-matrix-2}--\er{Lin-Par-eq-matrix-4} can be deduced 
from~\er{parametric-linear-tetrahedron-eqs} similarly.
\end{proof}
\begin{remark}
\label{rpequiv}
One can check that equations~\eqref{pmteq} are equivalent 
to~\eqref{Lin-Par-eq-da-ad}--\eqref{Par-Lin-eq-l-ll},
up to permutations of the parameters.
Thus, equations~\eqref{Lin-Par-eq-da-ad}--\eqref{Par-Lin-eq-l-ll}
can be replaced by equations~\eqref{pmteq}, 
which have more clear structure.
\end{remark}

\begin{proposition}
\lb{pipm}
For any vector space $V$,
the set of linear parametric tetrahedron maps~\er{tabg} 
is invariant under the following transformations
\begin{gather}
\lb{pcsym}
\begin{pmatrix}
         A_{\alpha \beta \gamma}&B_{\alpha \beta \gamma}&C_{\alpha \beta \gamma}  \\
         D_{\alpha \beta \gamma}&E_{\alpha \beta \gamma}&F_{\alpha \beta \gamma}  \\
         K_{\alpha \beta \gamma}&L_{\alpha \beta \gamma}&M_{\alpha \beta \gamma}
\end{pmatrix}\mapsto
\begin{pmatrix}
M_{\gamma \beta \alpha}&L_{\gamma \beta \alpha}&K_{\gamma \beta \alpha}  \\
         F_{\gamma \beta \alpha}&E_{\gamma \beta \alpha}&D_{\gamma \beta \alpha}  \\
         C_{\gamma \beta \alpha}&B_{\gamma \beta \alpha}&A_{\gamma \beta \alpha}
\end{pmatrix},\\
\lb{pchsn}
\begin{pmatrix}
         A_{\alpha \beta \gamma}&B_{\alpha \beta \gamma}&C_{\alpha \beta \gamma}  \\
         D_{\alpha \beta \gamma}&E_{\alpha \beta \gamma}&F_{\alpha \beta \gamma}  \\
         K_{\alpha \beta \gamma}&L_{\alpha \beta \gamma}&M_{\alpha \beta \gamma}
\end{pmatrix}\mapsto
\begin{pmatrix}
         -A_{\alpha \beta \gamma}&B_{\alpha \beta \gamma}&-C_{\alpha \beta \gamma}  \\
         D_{\alpha \beta \gamma}&-E_{\alpha \beta \gamma}&F_{\alpha \beta \gamma}  \\
         -K_{\alpha \beta \gamma}&L_{\alpha \beta \gamma}&-M_{\alpha \beta \gamma}
\end{pmatrix}.
\end{gather}

Let $V=\mathbb{K}^n$ for some $n\in\zsp$.
Then 
$A_{\alpha\beta\gamma}$, $B_{\alpha\beta\gamma}$, $C_{\alpha\beta\gamma}$,
$D_{\alpha\beta\gamma}$, $E_{\alpha\beta\gamma}$, $F_{\alpha\beta\gamma}$,
$K_{\alpha\beta\gamma}$, $L_{\alpha\beta\gamma}$, $M_{\alpha\beta\gamma}$ 
in~\eqref{tabg} are $n\times n$ matrices.
In this case, the set of linear parametric tetrahedron maps~\eqref{tabg} 
is invariant also under the transformation
\begin{gather}
\lb{ptrtr}
\begin{pmatrix}
         A_{\alpha \beta \gamma}&B_{\alpha \beta \gamma}&C_{\alpha \beta \gamma}  \\
         D_{\alpha \beta \gamma}&E_{\alpha \beta \gamma}&F_{\alpha \beta \gamma}  \\
         K_{\alpha \beta \gamma}&L_{\alpha \beta \gamma}&M_{\alpha \beta \gamma}
\end{pmatrix}\mapsto
\begin{pmatrix}
         A_{\alpha \beta \gamma}&B_{\alpha \beta \gamma}&C_{\alpha \beta \gamma}  \\
         D_{\alpha \beta \gamma}&E_{\alpha \beta \gamma}&F_{\alpha \beta \gamma}  \\
         K_{\alpha \beta \gamma}&L_{\alpha \beta \gamma}&M_{\alpha \beta \gamma}
 \end{pmatrix}^\ts =
\begin{pmatrix}
         A_{\alpha \beta \gamma}^\ts&D_{\alpha \beta \gamma}^\ts&K_{\alpha \beta \gamma}^\ts  \\
         B_{\alpha \beta \gamma}^\ts&E_{\alpha \beta \gamma}^\ts&L_{\alpha \beta \gamma}^\ts  \\
         C_{\alpha \beta \gamma}^\ts&F_{\alpha \beta \gamma}^\ts&M_{\alpha \beta \gamma}^\ts
\end{pmatrix}.
\end{gather}
\end{proposition}
\begin{proof}
The statement about the transformation~\er{pcsym} 
follows from Proposition~\ref{pp13} with $W=\Omega\times V$.

The case of the transformation~\er{pchsn} follows from Proposition~\ref{pist},
if we take $W=\Omega\times V$ and consider the map 
$$
\sigma\cl \Omega\times V\to \Omega\times V,\qqquad 
\sigma(\xi,v)=(\xi,-v),\qquad\xi\in\Omega,\qquad v\in V.
$$

To prove the statement about the transformation~\er{ptrtr},
one can apply the transpose operation 
to both sides of the parametric tetrahedron equation~\er{Par-Tetrahedron-eq}
for $T_{\alpha \beta \gamma}$ given by~\er{tabg}.
\end{proof}

\begin{proposition}\label{lin-par-family-abc}
Let $T_{1,\alpha\beta\gamma}$, $T_{2,\alpha\beta\gamma}$ be linear parametric tetrahedron maps of the form
\begin{equation*}
    T_{1,\alpha\beta\gamma}=\begin{pmatrix}
     A_{\alpha\beta\gamma}&B_{\alpha \beta \gamma}&0  \\
     D_{\alpha \beta \gamma}&E_{\alpha \beta \gamma}&0  \\
     0&0&M_{\alpha \beta \gamma}
\end{pmatrix},\qqquad
    T_{2,\alpha\beta\gamma}=\begin{pmatrix}
     \tilde{A}_{\alpha\beta\gamma}&0&0  \\
     0&\tilde{E}_{\alpha\beta\gamma}&\tilde{F}_{\alpha\beta\gamma}\\
     0&\tilde{L}_{\alpha\beta\gamma}&\tilde{M}_{\alpha\beta\gamma}
\end{pmatrix}.
\end{equation*}
Let $l,m\in\mathbb{K}$, $l\neq 0$. Then 
\begin{equation*}
    T_{1,\alpha\beta\gamma}^{l,m}=\begin{pmatrix}
     lA_{\alpha\beta\gamma}&B_{\alpha\beta\gamma}&0  \\
     D_{\alpha\beta\gamma}&l^{-1}E_{\alpha\beta\gamma}&0  \\
     0&0&mM_{\alpha\beta\gamma}
\end{pmatrix},\qqquad
    T_{2,\alpha\beta\gamma}^{l,m}=\begin{pmatrix}
     m\tilde{A}_{\alpha\beta\gamma}&0&0  \\
     0&l\tilde{E}_{\alpha\beta\gamma}&\tilde{F}_{\alpha\beta\gamma}\\
     0&\tilde{L}_{\alpha\beta\gamma}&l^{-1}\tilde{M}_{\alpha\beta\gamma}
\end{pmatrix}
\end{equation*}
are linear parametric tetrahedron maps as well.
\end{proposition}
\begin{proof}
For each $i=1,2$, the fact that $T_{i,\alpha\beta\gamma}$ 
obeys equations~\er{parametric-linear-tetrahedron-eqs} 
implies that $T_{i,\alpha\beta\gamma}^{l,m}$ obeys these equations as well.
\end{proof}

\section{Differentials of Yang--Baxter and tetrahedron maps}
\label{sdybt}

In this section, when we consider maps of manifolds, 
we assume that they are either smooth, or complex-analytic, or rational,
so that the differential is defined for such a map.

Let $\md$ be a manifold. 
Consider the tangent bundle $\tau\cl T\md\to\md$. Then 
\begin{itemize}
	\item the bundle $\tau\times\tau\cl T\md\times T\md\to \md\times \md$ 
can be identified with the tangent bundle of the manifold $\md\times \md$,
\item the bundle $\tau\times\tau\times\tau\cl T\md\times T\md\times T\md\to \md\times \md\times \md$ 
can be identified with the tangent bundle of the manifold $\md\times \md\times \md$.
\end{itemize}
Using these identifications and the general procedure 
to define the differential of a map of manifolds, for any maps 
$$
Y\cl\md\times \md\to\md\times\md,\qqquad 
\mathbf{T}\cl\md\times \md\times\md\to\md\times\md\times\md
$$
we obtain the differentials
$$
\dft Y\cl T\md\times T\md\to T\md\times T\md,\qqquad
\dft\mathbf{T}\cl T\md\times T\md\times T\md\to T\md\times T\md\times T\md.
$$

\begin{remark}
\lb{rdl}
The maps $\dft Y$, $\dft\mathbf{T}$ are linear along the fibres 
of the bundle $T\md\to\md$ and, in general, 
are nonlinear with respect to (local) coordinates on the manifold~$\md$.
The explicit computation of $\dft Y$, $\dft\mathbf{T}$ 
in coordinates is presented below.
\end{remark}

We need the following well-known property of differentials.
\begin{lemma}
\lb{lcdiff}
Let $\md_1$, $\md_2$, $\md_3$ be manifolds.
Consider maps $f\cl\md_1\to\md_2$, $\,g\cl\md_2\to\md_3$ and their differentials 
$\dft f\cl T\md_1\to T\md_2$, $\,\dft g\cl T\md_2\to T\md_3$.

Then for the differential $\dft(g\circ f)\cl T\md_1\to T\md_3$ of the composition map 
$g\circ f\cl \md_1\to \md_3$ we have $\dft(g\circ f)=\dft g\circ\dft f$.
\end{lemma}

\begin{proposition}
\label{thdyb}
Let $\md$ be a manifold.
For any Yang--Baxter map $Y\cl\md\times \md\to \md\times\md$,
the differential $\dft Y\cl T\md\times T\md\to T\md\times T\md$
is a Yang--Baxter map of the manifold $T\md\times T\md$.
\end{proposition}
\begin{proof}
Consider the permutation maps
\begin{gather*}
P^{12}\cl \md\times\md\times\md\to\md\times\md\times\md,\qquad 
P^{12}(a_1,a_2,a_3)=(a_2,a_1,a_3),\qquad a_i\in\md,\\
\tilde{P}^{12}\cl T\md\times T\md\times T\md\to T\md\times T\md\times T\md,\qquad 
\tilde{P}^{12}(b_1,b_2,b_3)=(b_2,b_1,b_3),\qquad b_i\in T\md,
\end{gather*}
and the identity maps $\id_{\md}\cl\md\to\md$, $\,\id_{T\md}\cl T\md\to T\md$.

$Y\cl\md\times \md\to \md\times\md$ obeys the Yang--Baxter equation
\begin{equation}
\lb{yby}
Y^{12}\circ Y^{13}\circ Y^{23}=Y^{23}\circ Y^{13}\circ Y^{12},
\end{equation}
where the maps $Y^{12},Y^{13},Y^{23}\cl\md\times\md\times\md\to\md\times\md\times\md$ 
can be described as follows
\begin{gather}
\lb{y12y23}
Y^{12}=Y\times\id_{\md},\qqquad Y^{23}=\id_{\md}\times Y,\\
\lb{y13}
Y^{13}=P^{12}\circ (\id_{\md}\times Y)\circ P^{12}.
\end{gather}

We need to prove that the map $\dft Y\cl T\md\times T\md\to T\md\times T\md$
obeys the Yang--Baxter equation
\begin{equation}
\lb{ybedy}
(\dft Y)^{12}\circ (\dft Y)^{13}\circ (\dft Y)^{23}=(\dft Y)^{23}\circ (\dft Y)^{13}\circ (\dft Y)^{12},
\end{equation}
where
\begin{gather*}
(\dft Y)^{12}=\dft Y\times\id_{T\md},\qqquad (\dft Y)^{23}=\id_{T\md}\times\dft Y,\qqquad
(\dft Y)^{13}=\tilde{P}^{12}\circ (\id_{T\md}\times\dft Y)\circ\tilde{P}^{12}.
\end{gather*}
By Lemma~\ref{lcdiff}, 
\begin{gather}
\lb{dddd}
\dft(Y^{12}\circ Y^{13}\circ Y^{23})=
\dft(Y^{12})\circ\dft(Y^{13})\circ\dft(Y^{23}),\quad
\dft(Y^{23}\circ Y^{13}\circ Y^{12})=
\dft(Y^{23})\circ\dft(Y^{13})\circ\dft(Y^{12}).
\end{gather}
Taking the differential of~\er{yby} and using~\er{dddd}, we obtain
\begin{gather}
\lb{dd}
\dft(Y^{12})\circ\dft(Y^{13})\circ\dft(Y^{23})=
\dft(Y^{23})\circ\dft(Y^{13})\circ\dft(Y^{12}).
\end{gather}

From~\er{y12y23} one derives 
\begin{gather*}
\dft(Y^{12})=\dft(Y\times\id_{\md})=\dft Y\times\id_{T\md}=(\dft Y)^{12},\quad
\dft(Y^{23})=\dft(\id_{\md}\times Y)=\id_{T\md}\times\dft Y=(\dft Y)^{23}.
\end{gather*}
Using Lemma~\ref{lcdiff} and the relation $\dft(P^{12})=\tilde{P}^{12}$, from~\er{y13} we obtain
$$
\dft(Y^{13})=\dft\big(P^{12}\circ (\id_{\md}\times Y)\circ P^{12}\big)=
\dft(P^{12})\circ\dft(\id_{\md}\times Y)\circ\dft(P^{12})
=\tilde{P}^{12}\circ (\id_{T\md}\times\dft Y)\circ\tilde{P}^{12}=(\dft Y)^{13}.
$$
Thus, we have 
\begin{gather}
\lb{dyij}
\dft(Y^{12})=(\dft Y)^{12},\qqquad \dft(Y^{13})=(\dft Y)^{13},\qqquad \dft(Y^{23})=(\dft Y)^{23}.
\end{gather}
Substituting~\er{dyij} in~\er{dd}, one obtains~\er{ybedy}.
\end{proof}
The statement of Proposition~\ref{thdyb} was used (without proof) in~\cite{BIKRP}.

\begin{proposition}
\label{thdybtet}
Let $\md$ be a manifold.
For any tetrahedron map 
$\mathbf{T}\cl\md\times \md\times\md\to\md\times\md\times\md$, 
the differential 
$$
\dft\mathbf{T}\cl T\md\times T\md\times T\md\to T\md\times T\md\times T\md
$$ 
is a tetrahedron map of the manifold $T\md\times T\md\times T\md$.
\end{proposition}
\begin{proof}
The map $\mathbf{T}\cl\md\times \md\times\md\to\md\times\md\times\md$
obeys the the tetrahedron equation
\begin{equation}\label{bftetr}
\mathbf{T}^{123}\circ \mathbf{T}^{145} \circ\mathbf{T}^{246}\circ \mathbf{T}^{356}
=\mathbf{T}^{356}\circ \mathbf{T}^{246}\circ\mathbf{T}^{145}\circ \mathbf{T}^{123}.
\end{equation}
We need to show that 
$\dft\mathbf{T}\cl T\md\times T\md\times T\md\to T\md\times T\md\times T\md$
satisfies the tetrahedron equation
\begin{equation}\label{dbftetr}
(\dft\mathbf{T})^{123}\circ(\dft\mathbf{T})^{145}\circ
(\dft\mathbf{T})^{246}\circ(\dft\mathbf{T})^{356}
=(\dft\mathbf{T})^{356}\circ(\dft\mathbf{T})^{246}
\circ(\dft\mathbf{T})^{145}\circ(\dft\mathbf{T})^{123}.
\end{equation}
Here for $1\le i<j<k\le 6$ the map $(\dft\mathbf{T})^{ijk}\cl (T\md)^6\to (T\md)^6$ is constructed 
from~$\dft\mathbf{T}$ similarly to the construction of~$\mathbf{T}^{ijk}$ from~$\mathbf{T}$.

Taking the differential of~\er{bftetr} and using Lemma~\ref{lcdiff}, we derive
\begin{equation}\label{dt123}
\dft(\mathbf{T}^{123})\circ\dft(\mathbf{T}^{145})\circ\dft(\mathbf{T}^{246})
\circ\dft(\mathbf{T}^{356})
=\dft(\mathbf{T}^{356})\circ\dft(\mathbf{T}^{246})\circ\dft(\mathbf{T}^{145})
\circ\dft(\mathbf{T}^{123}).
\end{equation}
Similarly to obtaining~\er{dyij}, one can show the following
\begin{gather}
\lb{dtdt}
\dft(\mathbf{T}^{123})=(\dft\mathbf{T})^{123},\qquad
\dft(\mathbf{T}^{145})=(\dft\mathbf{T})^{145},\qquad
\dft(\mathbf{T}^{246})=(\dft\mathbf{T})^{246},\qquad
\dft(\mathbf{T}^{356})=(\dft\mathbf{T})^{356}.
\end{gather}
For example, let us prove $\dft(\mathbf{T}^{246})=(\dft\mathbf{T})^{246}$.
For any $i,j\in\{1,2,\dots,6\}$, $i<j$, let $P^{ij}\cl(\md)^6\to(\md)^6$ 
be the permutation map which interchanges 
the $i$th and $j$th factors of the Cartesian product~$(\md)^6$.
Then 
$$
\tilde{P}^{ij}=\dft(P^{ij})\cl(T\md)^6\to(T\md)^6
$$
is the permutation map of the same type for the Cartesian product~$(T\md)^6$.
We have
\begin{gather*}
\mathbf{T}^{246}=P^{45}\circ P^{23}\circ P^{34}\circ 
(\id_{(\md)^3}\times\mathbf{T})\circ P^{34}\circ P^{23}\circ P^{45},\\
(\dft\mathbf{T})^{246}=\tilde{P}^{45}\circ \tilde{P}^{23}\circ \tilde{P}^{34}\circ 
(\id_{(T\md)^3}\times\dft\mathbf{T})\circ \tilde{P}^{34}\circ \tilde{P}^{23}\circ \tilde{P}^{45}.
\end{gather*}
Using these formulas, Lemma~\ref{lcdiff}, 
and the relation $\dft(P^{ij})=\tilde{P}^{ij}$, we obtain
\begin{multline*}
\dft(\mathbf{T}^{246})=\dft\big(P^{45}\circ P^{23}\circ P^{34}\circ 
(\id_{(\md)^3}\times\mathbf{T})\circ P^{34}\circ P^{23}\circ P^{45}\big)=\\
=\dft(P^{45})\circ\dft(P^{23})\circ\dft(P^{34})\circ 
\dft(\id_{(\md)^3}\times\mathbf{T})\circ\dft(P^{34})\circ\dft(P^{23})\circ\dft(P^{45})=\\
=\tilde{P}^{45}\circ \tilde{P}^{23}\circ \tilde{P}^{34}\circ 
(\id_{(T\md)^3}\times\dft\mathbf{T})\circ \tilde{P}^{34}\circ \tilde{P}^{23}\circ \tilde{P}^{45}=
(\dft\mathbf{T})^{246}.
\end{multline*}

Similarly, one can prove all of~\er{dtdt}. 
Substituting~\er{dtdt} in~\er{dt123}, one obtains~\er{dbftetr}.
\end{proof}

\begin{corollary}
\lb{ctaaa}
Consider a manifold $\md$, a tetrahedron map
$\mathbf{T}\cl\md\times \md\times\md\to\md\times\md\times\md$, and 
its differential
$$
\dft\mathbf{T}\cl T\md\times T\md\times T\md\to T\md\times T\md\times T\md.
$$

Let $a\in\md$ such that $\mathbf{T}\big((a,a,a)\big)=(a,a,a)$.
Consider the tangent space $T_a\md\subset T\md$ at the point $a\in\md$.
Then we have
\begin{gather}
\lb{dts}
\dft\mathbf{T}(T_a\md\times T_a\md\times T_a\md)\subset
T_a\md\times T_a\md\times T_a\md\subset T\md\times T\md\times T\md,
\end{gather}
and the map 
\begin{gather}
\lb{dtaa}
\dft\mathbf{T}\big|_{(a,a,a)}\cl
T_a\md\times T_a\md\times T_a\md\to T_a\md\times T_a\md\times T_a\md
\end{gather}
is a linear tetrahedron map. Here $\dft\mathbf{T}\big|_{(a,a,a)}$ is 
the restriction of the map $\dft\mathbf{T}$ to $T_a\md\times T_a\md\times T_a\md$. 
\end{corollary}
\begin{proof}
The property $\mathbf{T}\big((a,a,a)\big)=(a,a,a)$ and the definition 
of the differential imply~\er{dts} and the fact that the map~\er{dtaa} is linear.

By Proposition~\ref{thdybtet}, the differential $\dft\mathbf{T}$ is a tetrahedron map.
Therefore, its restriction $\dft\mathbf{T}\big|_{(a,a,a)}$ to 
$T_a\md\times T_a\md\times T_a\md$ is a tetrahedron map as well.
\end{proof}

\begin{remark}
\lb{rlappr}
The definition of the differential implies that 
the linear tetrahedron map $\dft\mathbf{T}\big|_{(a,a,a)}$ 
described in Corollary~\ref{ctaaa} can be regarded as a linear approximation of
the nonlinear tetrahedron map~$\mathbf{T}$ at the point $(a,a,a)\in\md\times\md\times\md$.
Explicit examples of $\dft\mathbf{T}\big|_{(a,a,a)}$ are 
presented in Examples~\ref{eeltr},~\ref{edmh}.
\end{remark}

Let $W$ be a set.
Three maps $G,H,Q\cl W\times W\to W\times W$ are called 
\emph{entwining Yang--Baxter maps} if they satisfy
\begin{gather}
\lb{ghqqhg}
G^{12}\circ H^{13}\circ Q^{23}=Q^{23}\circ H^{13}\circ G^{12}
\end{gather}
(see, e.g.,~\cite{KoulPap11,Pavlos2019,KP2019}).
The maps $G^{12},H^{13},Q^{23}\cl W\times W\times W\to W\times W\times W$
are constructed from $G$, $H$, $Q$ in the standard way.
One has~\er{ghqqhg},
but the maps $G$, $H$, $Q$ individually do not necessarily satisfy the Yang--Baxter equation.
Similarly to Proposition~\ref{thdyb}, one can prove the following.
\begin{proposition}
\label{pentw}
Let $\md$ be a manifold.
For any entwining Yang--Baxter maps $G,H,Q\cl \md\times\md\to\md\times\md$,
the differentials $\dft G,\dft H,\dft Q\cl T\md\times T\md\to T\md\times T\md$
are entwining Yang--Baxter maps of the manifold $T\md\times T\md$.
\end{proposition}

Let $n\in\zsp$. Let $\md$ be an $n$-dimensional manifold 
with (local) coordinates $x_1,\dots,x_n$.
Then $\dim T\md =2n$, and we have the (local) coordinates $x_1,\dots,x_n,X_1,\dots,X_n$ 
on the manifold~$T\md$, where $X_i$ corresponds to the differential $\dft x_i$,
which can be regarded as a function on~$T\md$. 
(Thus, the functions $X_1,\dots,X_n$ are linear along the fibres of the bundle $T\md\to\md$.)

To study maps of the form 
$$
\md\times \md\times\md\to\md\times\md\times\md,\qqquad
T\md\times T\md\times T\md\to T\md\times T\md\times T\md,
$$
we consider 
\begin{itemize}
	\item $6$ copies of the manifold $\md$ with coordinate systems
\begin{gather*}
(x_1,\dots,x_n),\qquad(y_1,\dots,y_n),\qquad(z_1,\dots,z_n),\qquad
(\tilde{x}_1,\dots,\tilde{x}_n),\qquad(\tilde{y}_1,\dots,\tilde{y}_n),\qquad(\tilde{z}_1,\dots,\tilde{z}_n),
\end{gather*}
\item $6$ copies of the manifold $T\md$ with coordinate systems
\begin{gather*}
(x_1,\dots,x_n,X_1,\dots,X_n),\qqquad(y_1,\dots,y_n,Y_1,\dots,Y_n),\qqquad(z_1,\dots,z_n,Z_1,\dots,Z_n),\\
(\tilde{x}_1,\dots,\tilde{x}_n,\tilde{X}_1,\dots,\tilde{X}_n),\qqquad
(\tilde{y}_1,\dots,\tilde{y}_n,\tilde{Y}_1,\dots,\tilde{Y}_n),\qqquad
(\tilde{z}_1,\dots,\tilde{z}_n,\tilde{Z}_1,\dots,\tilde{Z}_n).
\end{gather*}
\end{itemize}
Here, for each $i=1,\dots,n$, the functions $X_i$, $Y_i$, $Z_i$, $\tilde{X}_i$, $\tilde{Y}_i$, $\tilde{Z}_i$
correspond to the differentials $\dft x_i$, $\dft y_i$, $\dft z_i$, $\dft\tilde{x}_i$, $\dft\tilde{y}_i$, $\dft\tilde{z}_i$.
Below we use the following  notation
\begin{gather*}
x=(x_1,\dots,x_n),\qquad y=(y_1,\dots,y_n),\qquad z=(z_1,\dots,z_n),\\
X=(X_1,\dots,X_n),\qqquad Y=(Y_1,\dots,Y_n),\qqquad Z=(Z_1,\dots,Z_n),\\
\tilde{x}=(\tilde{x}_1,\dots,\tilde{x}_n),\qquad
\tilde{y}=(\tilde{y}_1,\dots,\tilde{y}_n),\qquad\tilde{z}=(\tilde{z}_1,\dots,\tilde{z}_n),\\
\tilde{X}=(\tilde{X}_1,\dots,\tilde{X}_n),\qqquad
\tilde{Y}=(\tilde{Y}_1,\dots,\tilde{Y}_n),\qqquad
\tilde{Z}=(\tilde{Z}_1,\dots,\tilde{Z}_n).
\end{gather*}

Consider a tetrahedron map
\begin{gather*}
\mathbf{T}\cl\md\times \md\times\md\to\md\times\md\times\md,\qqquad
(x,y,z)\mapsto(\tilde{x},\tilde{y},\tilde{z}),\\
\tilde{x}_i=f_i(x,y,z),\qqquad
\tilde{y}_i=g_i(x,y,z),\qqquad
\tilde{z}_i=h_i(x,y,z),\qqquad i=1,\dots,n.
\end{gather*}
Its differential is the following tetrahedron map 
\begin{gather*}
\dft\mathbf{T}\cl T\md\times T\md\times T\md\to T\md\times T\md\times T\md,\qqquad
(x,X,y,Y,z,Z)\mapsto(\tilde{x},\tilde{X},\tilde{y},\tilde{Y},\tilde{z},\tilde{Z}),\\
\tilde{x}_i=f_i(x,y,z),\qqquad
\tilde{y}_i=g_i(x,y,z),\qqquad
\tilde{z}_i=h_i(x,y,z),\qqquad i=1,\dots,n,\\
\tilde{X}_i=\sum_{j=1}^n\Big(\frac{\partial f_i(x,y,z)}{\partial x_j}X_j
+\frac{\partial f_i(x,y,z)}{\partial y_j}Y_j+\frac{\partial f_i(x,y,z)}{\partial z_j}Z_j\Big),\\
\tilde{Y}_i=\sum_{j=1}^n\Big(\frac{\partial g_i(x,y,z)}{\partial x_j}X_j
+\frac{\partial g_i(x,y,z)}{\partial y_j}Y_j+\frac{\partial g_i(x,y,z)}{\partial z_j}Z_j\Big),\\
\tilde{Z}_i=\sum_{j=1}^n\Big(\frac{\partial h_i(x,y,z)}{\partial x_j}X_j
+\frac{\partial h_i(x,y,z)}{\partial y_j}Y_j+\frac{\partial h_i(x,y,z)}{\partial z_j}Z_j\Big).
\end{gather*}
Note that the map $\dft\mathbf{T}$ is linear with respect to $X$, $Y$, $Z$ and, in general, is nonlinear 
with respect to $x$, $y$,~$z$.

\begin{example}
\lb{eeltr}
Let $n=\dim\md=1$.
Consider the well-known electric network transformation
\begin{gather}
\lb{elnt1}
\mathbf{T}\cl\md\times \md\times\md\to\md\times\md\times\md,\qqquad
(x,y,z)\mapsto (\tilde{x},\tilde{y},\tilde{z}),\\
\lb{elnt2}
\tilde{x}=\frac{x y}{x + z + x y z},\qqquad
\tilde{y}=x + z + x y z,\qqquad
\tilde{z}=\frac{y z}{x y z+x+z},
\end{gather}
which is a tetrahedron map~\cite{Sergeev,Kashaev-Sergeev}.
Its differential is the following tetrahedron map 
\begin{gather}
\lb{difelt}
\dft\mathbf{T}\cl T\md\times T\md\times T\md\to T\md\times T\md\times T\md,\qqquad
(x,X,y,Y,z,Z)\mapsto(\tilde{x},\tilde{X},\tilde{y},\tilde{Y},\tilde{z},\tilde{Z}),\\
\notag
\tilde{x}=\frac{x y}{x + z + x y z},\qqquad
\tilde{y}=x + z + x y z,\qqquad
\tilde{z}=\frac{y z}{x y z+x+z},\\
\lb{delnt1}
\tilde{X}=\frac{- x y (1 + x y)Z +  y zX +  x (x + z)Y}{(x y z+x+z)^2},\qqquad
\tilde{Y}=X + Z + x y Z+ x zY + y zX,\\
\lb{delnt2}
\tilde{Z}=\frac{- y z (y z+1)X+z (x+z)Y+x yZ}{(x y z+x+z)^2}.
\end{gather}

We assume that $x,y,z,\tilde{x},\tilde{y},\tilde{z}$ take values in~$\mathbb{C}$, 
so $\md$ is a complex manifold.
Consider $\mathrm{i}=\sqrt{-1}\in\mathbb{C}$.
Formulas~\er{elnt1},~\er{elnt2} imply 
$\mathbf{T}\big((\mathrm{i},\mathrm{i},\mathrm{i})\big)=
(\mathrm{i},\mathrm{i},\mathrm{i})$.

Let $a=\mathrm{i}$.
The coordinate system on~$\md$ gives the isomorphism $T_a\md\cong\mathbb{C}$.
By Corollary~\ref{ctaaa}, we obtain the linear tetrahedron map 
$\dft\mathbf{T}\big|_{(\mathrm{i},\mathrm{i},\mathrm{i})}\cl 
\mathbb{C}^3\to\mathbb{C}^3$. To compute it, 
we substitute $x=y=z=\mathrm{i}$ in~\er{delnt1},~\er{delnt2} and derive 
the linear map
\begin{gather}
\lb{xtx}
\begin{pmatrix}
     X  \\
     Y  \\
     Z
\end{pmatrix}\mapsto
\begin{pmatrix}
     \tilde{X}  \\
     \tilde{Y}  \\
     \tilde{Z}
\end{pmatrix}=
\begin{pmatrix}
     X+2Y  \\
     -Y  \\
     2Y+Z
\end{pmatrix}
\end{gather}
with the matrix 
$\begin{pmatrix}
     A&B&C  \\
     D&E&F  \\
     K&L&M
\end{pmatrix}=
\begin{pmatrix}
     1&2&0  \\
     0&-1&0  \\
     0&2&1
\end{pmatrix}$,
which is of the form~\er{hietm} for $\ah=1$, $\bh=-1$, $\ch=1$.

The linear tetrahedron map~\er{xtx} is a linear approximation 
of the map~\er{elnt1},~\er{elnt2} at the point $(\mathrm{i},\mathrm{i},\mathrm{i})$
in the following sense. We have
\begin{gather*}
\mathbf{T}\big((\mathrm{i}+\varepsilon X,\,
\mathrm{i}+\varepsilon Y,\,\mathrm{i}+\varepsilon Z)\big)=
\big(\mathrm{i}+\varepsilon(X+2Y)+\mathcal{O}(\varepsilon^2),\,
\mathrm{i}-\varepsilon Y+\mathcal{O}(\varepsilon^2),\,
\mathrm{i}+\varepsilon(2Y+Z)+\mathcal{O}(\varepsilon^2)\big).
\end{gather*}

\end{example}

\begin{example}
\lb{eknpt}
Let $n=\dim\md=2$.
Consider the Kassotakis--Nieszporski--Papageorgiou--Tongas map 
(map (33) in~\cite{Kassotakis-Tetrahedron})
\begin{gather*}
\mathbf{T}\cl\md\times \md\times\md\to\md\times\md\times\md,\qqquad
(x_1,x_2,y_1,y_2,z_1,z_2)\mapsto 
(\tilde{x}_1,\tilde{x}_2,\tilde{y}_1,\tilde{y}_2,\tilde{z}_1,\tilde{z}_2),\\
\tilde{x}_1=\frac{y_1+x_1 z_1}{z_1},\,\quad
\tilde{x}_2=y_2 z_1,\,\quad
\tilde{y}_1=x_1 z_1,\,\quad 
\tilde{y}_2=\frac{(y_1+x_1 z_1) z_2}{z_1},\,\quad 
\tilde{z}_1=\frac{y_1 z_1}{y_1+x_1 z_1},\,\quad 
\tilde{z}_2=\frac{y_2}{x_1}.
\end{gather*}
Its differential is the following tetrahedron map 
\begin{gather}
\notag
\dft\mathbf{T}\cl T\md\times T\md\times T\md\to T\md\times T\md\times T\md,\\
\lb{difkm}
(x_1,x_2,X_1,X_2,y_1,y_2,Y_1,Y_2,z_1,z_2,Z_1,Z_2)\mapsto
(\tilde{x}_1,\tilde{x}_2,\tilde{X}_1,\tilde{X}_2,\tilde{y}_1,\tilde{y}_2,
\tilde{Y}_1,\tilde{Y}_2,\tilde{z}_1,\tilde{z}_2,\tilde{Z}_1,\tilde{Z}_2),\\
\notag
\tilde{x}_1=\frac{y_1+x_1 z_1}{z_1},\,\quad
\tilde{x}_2=y_2 z_1,\,\quad
\tilde{y}_1=x_1 z_1,\,\quad 
\tilde{y}_2=\frac{(y_1+x_1 z_1) z_2}{z_1},\,\quad 
\tilde{z}_1=\frac{y_1 z_1}{y_1+x_1 z_1},\,\quad 
\tilde{z}_2=\frac{y_2}{x_1},\\
\notag
\tilde{X}_1=X_1 + \frac{1}{z_1}Y_1-\frac{y_1}{z_1^2}Z_1,\qqquad
\tilde{X}_2=z_1Y_2+y_2Z_1,\\
\notag
\tilde{Y}_1=z_1X_1+x_1Z_1,\qqquad
\tilde{Y}_2=z_2X_1+\frac{z_2}{z_1}Y_1-\frac{y_1 z_2}{z_1^2}Z_1
+\frac{y_1 + x_1 z_1}{z_1}Z_2,\\ 
\notag
\tilde{Z}_1=-\frac{y_1z_1^2}{(y_1 + x_1 z_1)^2}X_1+
\frac{x_1z_1^2}{(y_1 + x_1 z_1)^2}Y_1+
\frac{y_1^2}{(y_1 + x_1 z_1)^2}Z_1,\qqquad
\tilde{Z}_2=-\frac{ y_2}{x_1^2}X_1+\frac{1}{x_1}Y_2.
\end{gather}
\end{example}

\begin{example}
\lb{edmh}
Let $n=\dim\md=2$.
In a study of soliton solutions of vector KP equations,
Dimakis and M\"uller-Hoissen~\cite{Dimakis} 
constructed the tetrahedron map 
\begin{gather}
\lb{dmt}
\mathbf{T}\cl\md\times \md\times\md\to\md\times\md\times\md,\qqquad
(x_1,x_2,y_1,y_2,z_1,z_2)\mapsto 
(\tilde{x}_1,\tilde{x}_2,\tilde{y}_1,\tilde{y}_2,\tilde{z}_1,\tilde{z}_2),\\
\notag
\tilde{x}_1=y_1 \mathsf{C},\qqquad
\tilde{x}_2=\Big(y_1 - \frac{\mathsf{A}}{x_1}\Big) \mathsf{C},\qqquad
\tilde{y}_1=\frac{x_1}{\mathsf{C}},\qqquad 
\tilde{y}_2=1 - \mathsf{B},\\ 
\notag
\tilde{z}_1=\frac{z_1 y_1(x_1-x_2)}{\mathsf{A}},\qqquad 
\tilde{z}_2=1 - \frac{(1-y_2) (1-z_2)}{\mathsf{B}},\\
\notag
\mathsf{A} = y_2 z_1 x_1-y_2 x_1-z_1 x_2+x_1 y_1,\qqquad 
\mathsf{B} = y_2 z_2 x_1-y_2 x_1-z_2 x_2+1,\\
\notag
\mathsf{C} = \frac{\mathsf{A} \mathsf{B}-\mathsf{A} (1-y_2) (1-z_2) x_1-\mathsf{B} z_1(x_1-x_2)}{\mathsf{A} \mathsf{B}-\mathsf{A} (1-y_2) (1-z_2)-\mathsf{B} z_1 y_1(x_1-x_2)}.
\end{gather}

We have found the following invariants for this map
\begin{gather}
\lb{invdm}
I_1(x_1,x_2,y_1,y_2,z_1,z_2) = x_1 y_1,\qqquad
I_2(x_1,x_2,y_1,y_2,z_1,z_2) = (y_2-1) (z_2-1),\\
\lb{invdm3}
I_3(x_1,x_2,y_1,y_2,z_1,z_2) = (x_1 - x_2) (y_1 - y_2) (z_1 - z_2).
\end{gather}
That is, for $\tilde{x}_1,\tilde{x}_2,\tilde{y}_1,\tilde{y}_2,\tilde{z}_1,\tilde{z}_2$ given by the above formulas, one has 
$$
I_j(\tilde{x}_1,\tilde{x}_2,\tilde{y}_1,\tilde{y}_2,\tilde{z}_1,\tilde{z}_2)=I_j(x_1,x_2,y_1,y_2,z_1,z_2),
\qquad j=1,2,3.
$$
The invariants $I_1$, $I_2$, $I_3$ are functionally independent.

We assume that $x_i,y_i,z_i,\tilde{x}_i,\tilde{y}_i,\tilde{z}_i$ take values in~$\mathbb{C}$, 
so $\md$ is a complex manifold.
One can check that for any nonzero $\ct\in\mathbb{C}$ we have $\mathbf{T}\big((a,a,a)\big)=(a,a,a)$, 
where $a=(\ct,0)\in\md$. 
Therefore, by Corollary~\ref{ctaaa}, one obtains the linear tetrahedron map
\begin{gather*}
\dft\mathbf{T}\big|_{(a,a,a)}\cl
T_a\md\times T_a\md\times T_a\md\to T_a\md\times T_a\md\times T_a\md.
\end{gather*}

The coordinate system on~$\md$ gives the isomorphism $T_a\md\cong\mathbb{C}^2$, 
so we have $\dft\mathbf{T}\big|_{(a,a,a)}\cl\mathbb{C}^6\to\mathbb{C}^6$.
Computing $\dft\mathbf{T}$ and $\dft\mathbf{T}\big|_{(a,a,a)}$ for $a=(\ct,0)$, 
one derives that $\dft\mathbf{T}\big|_{(a,a,a)}$ is given by the matrix 
\begin{gather}
\lb{pmat2}
\begin{pmatrix} 
1 & 0 & \frac{\ct-1}{\ct} & \frac{(\ct-1)^2 (\ct+1)}{\ct} & -\frac{\ct-1}{\ct} & 1-\ct \\
 0 & 1 & 0 & 1-\ct & 0 & 0 \\
 0 & 0 & \frac{1}{\ct} & -\frac{(\ct-1)^2 (\ct+1)}{\ct} & \frac{\ct-1}{\ct} & \ct-1 \\
 0 & 0 & 0 & \ct & 0 & 0 \\
 0 & 0 & 0 & 1-\ct & 1 & 0 \\
 0 & 0 & 0 & 1-\ct & 0 & 1
\end{pmatrix},
\end{gather}
which coincides with~\er{abcp}.
Thus, the linear map $\dft\mathbf{T}\big|_{(a,a,a)}$ is 
of the form~\er{abcp},~\er{pabc}.

According to Remark~\ref{rlappr},
the linear tetrahedron map~\er{pmat2} is a linear approximation of
the nonlinear tetrahedron map~\er{dmt} 
at the point $(a,a,a)\in\md\times\md\times\md$ with $a=(\ct,0)$, $\ct\neq 0$.
\end{example}

\section{Parametric Yang--Baxter maps associated with matrix groups}
\label{sYBmvb}

Let $\gmd$ be a group and $\pw\in\zsp$. 
It is known that one has the following Yang--Baxter map
\begin{gather}
\label{fxyx}
\fmp\colon \gmd\times \gmd\to \gmd\times \gmd,\qquad
\fmp(x,y)=(x,x^\pw yx^{-\pw}),\qquad
x,y\in \gmd,
\end{gather}
(see, e.g.,~\cite{carter2006} and references therein).
For $\pw=1$ this map appeared in~\cite{Drin92}

Assume that $\mathbb{K}$ is either $\mathbb{R}$ or $\mathbb{C}$.
Let $n\in\zsp$ and consider the matrix group
$\gmd=\GL_n(\mathbb{K})\subset\mat_n(\mathbb{K})$.
Then $\gmd$ is a manifold, and for each $x\in\gmd=\GL_n(\mathbb{K})$ 
one has the tangent space $T_x\gmd\cong\mat_n(\mathbb{K})$.
Set $\mm=\mat_n(\mathbb{K})$.
The tangent bundle of the manifold~$\gmd$ 
can be identified with the trivial bundle $\gmd\times\mm\to\gmd$.

For $\gmd=\GL_n(\mathbb{K})$, the Yang--Baxter map~\eqref{fxyx} is an analytic diffeomorphism of 
the manifold $\gmd\times\gmd$.
The differential~$\dft\fmp$ of this diffeomorphism~$\fmp$ can be identified with the following map
\begin{gather}
\label{dfxyx}
\begin{gathered}
\dft\fmp\colon 
(\gmd\times\mm)\times(\gmd\times\mm)
\to
(\gmd\times\mm)\times(\gmd\times\mm),\\
\dft\fmp\big((x,M_1),(y,M_2)\big)=\left(\Big(x,M_1\Big),
\Big(x^\pw yx^{-\pw},\frac{\partial}{\partial\varepsilon}\Big|_{\varepsilon=0}
\big((x+\varepsilon M_1)^\pw(y+\varepsilon M_2)(x+\varepsilon M_1)^{-\pw}\big)\Big)\right),
\end{gathered}\\
\notag
x,y\in\gmd=\GL_n(\mathbb{K}),\qquad\quad M_1,M_2\in\mm=\mat_n(\mathbb{K}).
\end{gather}
By Proposition~\ref{thdyb},
since $\fmp$ is a Yang--Baxter map, its differential $\dft\fmp$ is a Yang--Baxter map as well.

Let $\Omega\subset\gmd$ be an abelian subgroup of~$\gmd$.
Denote by 
$\ybm\colon(\Omega\times\mm)\times(\Omega\times\mm)\to(\Omega\times\mm)\times(\Omega\times\mm)$ 
the restriction of the map $\dft\fmp$ to the subset 
$(\Omega\times\mm)\times(\Omega\times\mm)\subset(\gmd\times\mm)\times(\gmd\times\mm)$.
As $\dft\fmp$ is a Yang--Baxter map, $\ybm$ is a Yang--Baxter map as well.

Let $a,b\in\Omega$. Since $ab=ba$, computing~\eqref{dfxyx} for $x=a$ and $y=b$, we obtain
\begin{gather}
\label{ybmgl}
\begin{gathered}
\ybm\colon
\big(\Omega\times\mat_n(\mathbb{K})\big)\times\big(\Omega\times\mat_n(\mathbb{K})\big)\to
\big(\Omega\times\mat_n(\mathbb{K})\big)\times\big(\Omega\times\mat_n(\mathbb{K})\big),\\
\ybm\big((a,M_1),(b,M_2)\big)=\bigg((a,M_1),
\Big(b,a^\pw M_2a^{-\pw}+
\sum_{i=0}^{\pw-1}\Big(a^iM_1a^{-i-1}b-ba^iM_1a^{-i-1}\Big)\Big)\bigg).
\end{gathered}
\end{gather}
The Yang--Baxter map~\eqref{ybmgl} can be interpreted as 
the following linear parametric Yang--Baxter map
\begin{gather}
\label{pyblg}
\begin{gathered}
\mathbb{Y}_{a,b}\colon\mat_n(\mathbb{K})\times\mat_n(\mathbb{K})\to
\mat_n(\mathbb{K})\times\mat_n(\mathbb{K}),\\
\mathbb{Y}_{a,b}(M_1,M_2)=
\bigg(M_1,\,a^\pw M_2a^{-\pw}+
\sum_{i=0}^{\pw-1}\Big(a^iM_1a^{-i-1}b-ba^iM_1a^{-i-1}\Big)\bigg),
\end{gathered}
\end{gather}
with parameters $a,b\in\Omega$.
We need the following result from~\cite{BIKRP}.
\begin{proposition}[\cite{BIKRP}] 
\label{prl}
Let $V$ be a vector space over a field $\mathbb{K}$.
Consider a linear parametric Yang--Baxter map 
$\mathsf{Y}_{\alpha\beta}\cl V\times V\to V\times V$ given by 
the formula
\begin{equation*}
\mathsf{Y}_{\alpha\beta}\cl\begin{pmatrix}
    x\\
    y
\end{pmatrix}\mapsto\begin{pmatrix}
    u\\
    v
\end{pmatrix} =\begin{pmatrix}
     \mathsf{A}_{\alpha \beta}&\mathsf{B}_{\alpha \beta} \\
     \mathsf{C}_{\alpha \beta}& \mathsf{D}_{\alpha \beta}
\end{pmatrix}
\begin{pmatrix}
    x\\
    y
\end{pmatrix}, \qquad
\mathsf{A}_{\alpha\beta},\mathsf{B}_{\alpha \beta},
\mathsf{C}_{\alpha \beta},\mathsf{D}_{\alpha \beta}\in\End(V),\qquad x,y\in V.
\end{equation*}

Then, for any nonzero constant $l\in\mathbb{K}$, the map 
\begin{gather*}
\mathsf{Y}^l_{\alpha\beta}\cl V\times V\to V\times V,\qqquad
\mathsf{Y}^l_{\alpha\beta}\cl\begin{pmatrix}
    x\\
    y
\end{pmatrix}\mapsto\begin{pmatrix}
    u\\
    v
\end{pmatrix} =\begin{pmatrix}
     l\mathsf{A}_{\alpha \beta}&\mathsf{B}_{\alpha \beta} \\
     \mathsf{C}_{\alpha \beta}& l^{-1}\mathsf{D}_{\alpha \beta}
\end{pmatrix}
\begin{pmatrix}
    x\\
    y
\end{pmatrix},\qqquad x,y\in V,
\end{gather*}
is a parametric Yang--Baxter map as well.
\end{proposition}

Let $l\in\mathbb{K}$, $l\neq 0$.
Applying Proposition~\ref{prl} to the map~\eqref{pyblg}, we obtain the linear parametric Yang--Baxter map
\begin{gather}
\label{lpyblg}
\begin{gathered}
Y^l_{a,b}\colon\mat_n(\mathbb{K})\times\mat_n(\mathbb{K})\to
\mat_n(\mathbb{K})\times\mat_n(\mathbb{K}),\\
Y^l_{a,b}(M_1,M_2)=\bigg(lM_1,\,l^{-1}a^\pw M_2a^{-\pw}+
\sum_{i=0}^{\pw-1}\Big(a^iM_1a^{-i-1}b-ba^iM_1a^{-i-1}\Big)\bigg),
\end{gathered}\\
\notag
a,b\in\Omega,\qquad\quad\text{$\Omega$ is an abelian subgroup of $\GL_n(\mathbb{K})$}.
\end{gather}
In the above construction of~\eqref{lpyblg} we have assumed
that $\mathbb{K}$ is either $\mathbb{R}$ or $\mathbb{C}$, 
in order to use tangent spaces and differentials.
Now one can verify that~\eqref{lpyblg} is a parametric Yang--Baxter map for any field~$\mathbb{K}$.

For $\pw=1$ the maps~\er{ybmgl},~\er{pyblg},~\er{lpyblg} were presented in~\cite{BIKRP}.
For $\pw\ge 2$ the maps~\er{ybmgl},~\er{pyblg},~\er{lpyblg} are new.

Using Corollary~\ref{cptyb}, 
from the parametric Yang--Baxter map~\er{pyblg} we obtain the following parametric tetrahedron map
\begin{gather}
\label{tpyblg}
\begin{gathered}
\mathbb{T}_{a,b,c}\colon\mat_n(\mathbb{K})\times\mat_n(\mathbb{K})\times\mat_n(\mathbb{K})
\to\mat_n(\mathbb{K})\times\mat_n(\mathbb{K})\times\mat_n(\mathbb{K}),\\
\mathbb{T}_{a,b,c}(M_1,M_2,M_3)=
\bigg(M_1,\,a^\pw M_2a^{-\pw}+
\sum_{i=0}^{\pw-1}\Big(a^iM_1a^{-i-1}b-ba^iM_1a^{-i-1}\Big),\,M_3\bigg),
\end{gathered}
\end{gather}
with parameters $a,b,c\in\Omega$.

Let $l,m\in\mathbb{K}$, $l\neq 0$.
Applying Proposition~\ref{lin-par-family-abc} to the map~\eqref{tpyblg},
we derive the parametric tetrahedron map
\begin{gather}
\label{tlmm}
\begin{gathered}
\mathbb{T}^{l,m}_{a,b,c}\colon\mat_n(\mathbb{K})\times\mat_n(\mathbb{K})\times\mat_n(\mathbb{K})
\to\mat_n(\mathbb{K})\times\mat_n(\mathbb{K})\times\mat_n(\mathbb{K}),\qquad
a,b,c\in\Omega,\\
\mathbb{T}^{l,m}_{a,b,c}(M_1,M_2,M_3)=
\bigg(lM_1,\,l^{-1}a^\pw M_2a^{-\pw}+
\sum_{i=0}^{\pw-1}\Big(a^iM_1a^{-i-1}b-ba^iM_1a^{-i-1}\Big),\,mM_3\bigg).
\end{gathered}
\end{gather}
The tetrahedron map~\er{tlmm} carries almost the same information as
the Yang--Baxter map~\er{lpyblg}, 
but we present~\er{tlmm} for completeness.

\section{Conclusions}
\label{sconc}

In this paper we have presented a number of results on 
tetrahedron maps and Yang--Baxter maps.

In particular, in Sections~\ref{sltm},~\ref{slptm}
we have clarified the structure of the nonlinear algebraic relations which define 
linear (parametric) tetrahedron maps (with nonlinear dependence on parameters).
Using this result,
in Propositions~\ref{pinvt},~\ref{par-family},~\ref{pipm},~\ref{lin-par-family-abc}
we have presented 
several transformations which allow one to obtain new such maps from known ones.

Furthermore, in Section~\ref{sdybt} we have proved that the differential 
of a (nonlinear) tetrahedron map on a manifold is a tetrahedron map as well.
Similar results on the differentials of Yang--Baxter 
and entwining Yang--Baxter maps are also presented in Section~\ref{sdybt}.

In Remark~\ref{rlappr}, Corollary~\ref{ctaaa}, and Examples~\ref{eeltr},~\ref{edmh} 
we have shown how linear tetrahedron maps appear as linear approximations of nonlinear ones.

Using the obtained general results, we have constructed a number 
of new Yang--Baxter and tetrahedron maps.

Example~\ref{edmh} is devoted to the 
nonlinear tetrahedron map~\er{dmt} from~\cite{Dimakis}.
We have obtained (functionally independent) invariants~\er{invdm},~\er{invdm3} for it.
Furthermore, we have constructed a family of new linear tetrahedron maps~\er{pmat2},
which are linear approximations of the map~\er{dmt}.
The family of maps~\er{pmat2} depends on the parameter~$\ct\in\mathbb{C}$.

In Examples~\ref{eeltr},~\ref{eknpt},
computing the differentials of some tetrahedron maps 
from~\cite{Sergeev,Kashaev-Sergeev,Kassotakis-Tetrahedron},
we have obtained new tetrahedron maps~\er{difelt},~\er{difkm}.

Let $\mathbb{K}$ be a field. 
(For instance, one can take $\mathbb{K}=\mathbb{C}$ or $\mathbb{K}=\mathbb{R}$.)
In Section~\ref{sYBmvb}, for any nonzero $l\in\mathbb{K}$, $\,n,\pw\in\zsp$, 
and any abelian subgroup $\Omega\subset\GL_n(\mathbb{K})$,
we have obtained the parametric Yang--Baxter map~\er{lpyblg} with parameters $a,b\in\Omega$.
For $\pw\ge 2$ the map~\er{lpyblg} is new.
For $\pw=1$ it was presented in~\cite{BIKRP}.

Motivated by the results of this paper,
we suggest the following directions for future research:
\begin{itemize}
\item There are several methods in the literature for associating Yang--Baxter and tetrahedron maps to discrete lattice equations~\cite{Pavlos-Maciej, Pavlos-Maciej-2, Kassotakis-Tetrahedron, pap-Tongas}. It would be interesting to compare discrete integrable systems related 
to given Yang--Baxter and tetrahedron maps with discrete systems associated with the differentials of these maps.
\item Noncommutative versions and extensions of Yang--Baxter and tetrahedron maps have been recently of particular interest (see, e.g., \cite{Doliwa, Doliwa-Kashaev, GKM, Sokor-2020, Sokor-Kouloukas}). In particular, Yang--Baxter maps extended by means of Grassmann algebras 
were obtained in~\cite{GKM, Sokor-Sasha-2}.

One can try to prove that the differentials of Grassmann extented Yang--Baxter 
and tetrahedron maps are also solutions to the Grassmann extended Yang--Baxter and tetrahedron equations. Moreover, we propose 
to compare the relation between Yang--Baxter and tetrahedron maps and their differentials 
versus the relation between the former and the latter in the case of Grassmann algebras.

\item 
It is well known that B\"acklund transformations 
for integrable partial differential, differential-difference 
and difference-difference equations can often be constructed by means of chains of
Miura-type transformations (also called Miura maps).
There is a method to construct Miura-type transformations for differential-difference
equations from Darboux--Lax representations (DLRs) of such equations~\cite{Igon2016}, 
using Lie group actions associated with matrices from DLRs.
The method in~\cite{Igon2016} is applicable to a wide class of DLRs
and can be extended to difference-difference equations.
It is inspired by the results of~\cite{drin-sok85,igon2005} on 
Miura-type transformations for (1+1)-dimensional evolution partial differential equations (PDEs).

On the other hand, there are examples of Yang--Baxter maps~\cite{Sokor-Sasha,KP2019,mpw2016} 
and tetrahedron maps~\cite{KR} arising from matrix refactorisation problems 
for Darboux matrices corresponding to Darboux--B\"acklund transformations 
of Lax representations of integrable (1+1)-dimensional evolution PDEs.
Furthermore, it is well known that such a Darboux--B\"acklund transformation 
very often gives integrable differential-difference and difference-difference equations
(see, e.g.,~\cite{hjn-book,kmw,mpw2016}).
This suggests to study relations between Miura maps, Yang--Baxter maps, and tetrahedron maps.
\end{itemize}

\section*{Acknowledgements} 
We acknowledge support by the Russian Science Foundation (grant No. 20-71-10110).

We would like to thank A.V.~Mikhailov and D.V.~Talalaev for useful discussions.

\end{document}